\documentclass{article}
\usepackage[utf8]{inputenc}
\usepackage{fullpage}
\usepackage{amsmath,amsthm,amssymb}
\usepackage{amsfonts}
\usepackage{enumitem}
\usepackage{thmtools,thm-restate}

\usepackage[draft]{graphicx}
\usepackage{subcaption}
\usepackage[british]{babel}
\usepackage[autostyle]{csquotes}
\usepackage{mathtools}
\usepackage{relsize}
\usepackage{mathalpha}
\usepackage{braket}
\usepackage{xfrac}
\usepackage{xcolor}
\usepackage{hyperref}
\definecolor{illiniBlue}{HTML}{0072CE}
\definecolor{illiniOrange}{HTML}{E84A27}
\hypersetup{
  colorlinks=true,
  linkcolor=illiniBlue,
  citecolor=illiniOrange,
  urlcolor=illiniBlue
}

\usepackage{comment}

\addtocontents{toc}{\protect\hypersetup{linkcolor=black}}

\usepackage{mathpazo}
\usepackage{courier}

\usepackage{algorithm}
\usepackage{algpseudocode}

\usepackage{tikz}

\usepackage{authblk}

\makeatletter
\newcommand{\leqnos}{\tagsleft@true\let\veqno\@@leqno}
\newcommand{\reqnos}{\tagsleft@false\let\veqno\@@eqno}
\reqnos
\makeatother

\newcommand{\ketbra}[2]{\ket{#1}\!\bra{#2}}
\newcommand{\tr}[1]{\textnormal{tr}\left[#1\right]}

\newcommand{\eps}{\varepsilon}

\newcommand{\BC}{\mathcal{B}}

\newcommand{\DC}{\mathcal{D}}
\newcommand{\EC}{\mathcal{E}}

\newcommand{\HC}{\mathcal{H}}

\newcommand{\LC}{\mathcal{L}}
\newcommand{\MC}{\mathcal{M}}

\newcommand{\OC}{\mathcal{O}}
\newcommand{\PC}{\mathcal{P}}

\newcommand{\SC}{\mathcal{S}}

\newcommand{\UC}{\mathcal{U}}

\newcommand{\XC}{\mathcal{X}}
\newcommand{\YC}{\mathcal{Y}}

\newcommand{\Var}{{\rm Var}}

\newif\ifnotes\notestrue
\ifnotes
 \usepackage{color}
 \definecolor{mygrey}{gray}{0.50}
 \newcommand{\notename}[2]{{\footnotesize{\bf (#1:} {#2}{\bf ) }}}

\else
 
 \newcommand{\notename}[2]{{}}
 
\fi

\newtheorem{theorem}{Theorem}[section]
\newtheorem{definition}{Definition}

\newtheorem{openquestion}{Open Question}

\newtheorem{proposition}[theorem]{Proposition}
\newtheorem{lemma}[theorem]{Lemma}
\newtheorem{corollary}[theorem]{Corollary}

\numberwithin{equation}{section}

\input{macros.sty}

\title{Product testing with single-copy measurements}

\author[1,2]{Jacob Beckey}
\author[3]{Luke Coffman}
\author[4]{Ariel Shlosberg}
\author[5]{Louis Schatzki}
\author[1,2]{Felix Leditzky}
\affil[1]{Department of Mathematics, University of Illinois Urbana-Champaign}
\affil[2]{Illinois Quantum Information Science and Technology Center (IQUIST), University of Illinois Urbana-Champaign}
\affil[3]{Quantum Science and Engineering, Harvard University}
\affil[4]{Center for Quantum Information and Control (CQuIC), University of New Mexico}
\affil[5]{Department of Electrical and Computer Engineering, University of Illinois Urbana-Champaign}

\date{}

\begin{document}

\maketitle

\begin{abstract}
In this work, we study the sample complexity of two variants of product testing when restricted to single-copy measurements. In particular, we consider both bipartite product testing (i.e., does there exist at least one non-trivial cut across which the state is product) and multipartite product testing (i.e., is the state fully product across every cut). For the first variant, we prove an exponential lower bound on the sample complexity of any algorithm for this task which utilizes only single-copy measurements. When comparing this with known efficient algorithms that utilize multi-copy measurements, this establishes an exponential separation for this and several related entanglement learning tasks. For the second variant, we prove another sample lower bound that establishes a separation between single- and multi-copy strategies. To obtain our results, we prove a crucial technical lemma that gives a lower bound on the overlap between tensor products of permutation operators acting on subsystems of states that themselves carry a tensor structure. Finally, we provide an algorithm for multipartite product testing using only single-copy, local measurements, and we highlight several interesting open questions arising from this work.
\end{abstract}

\tableofcontents
\newpage

\section{Introduction}
In recent years, there have been many works proving sample complexity separations between multi-copy and single-copy measurement strategies~\cite{aharonov2022Quantum,bubeck2020Entanglement,chen2022Exponential,chen2024optimal, ye2025Exponential,noller2025infinite}. A well-known example is for the canonical problem of quantum state tomography (QST), where any single-copy QST algorithm must use at least $\Omega(d^3/\eps^2)$ samples~\cite{haah2017SampleOptimal}, while the seminal works of Refs.~\cite{haah2017SampleOptimal,odonnell2016Efficient} showed that there exists a multi-copy strategy using at most $O(d^2/\varepsilon^2)$ samples. These works are emblematic of a burgeoning research program striving to understand the power of quantum memory~\cite{chen2022Exponential,chen2024Tight,anshu2022Distributed,gong2024sample,hinsche2025SingleCopy,kim2025Fundamental, ye2025Exponential,noller2025infinite}, adaptivity~\cite{chen2022When,oufkir2023Adaptivity,chen2024Adaptivity}, and randomness~\cite{liu2024role} in quantum learning and testing. 

Even more dramatically, there are problems for which there exist exponential separations between multi- and single-copy strategies. For example, \cite{chen2022Exponential} proved that \textit{purity testing} on a $d$-dimensional quantum state with single copy measurements requires at least $\Omega(\sqrt{d})$ copies, while in the multi-copy setting, there exist simple algorithms (e.g., the two-copy SWAP test~\cite{barenco1997Stabilization,ekert2002Direct}) to solve this problem using at most $O(1)$ samples. Given that $d$ typically scales exponentially with the number of subsystems, we say such a result is an \textit{exponential separation} between learning with and without quantum memory. Such substantial separations are particularly exciting because they are already observable in experiments using just tens of qubits~\cite{huang2022Quantum}. 

Entanglement is a central object of study in quantum theory due to its foundational theoretical importance and its broad utility in practical applications~\cite{horodecki2009Quantum}; however, it is notoriously hard to quantify and characterize~\cite{guhne2009Entanglement}. A task with both practical and deep complexity-theoretic implication is that of \textit{product testing}, which was rigorously studied by Harrow and Montanaro in Ref.~\cite{harrow_testing_2010}. Given $T$ copies of a pure $n$-qudit state $\ket{\psi} \in (\mathbb{C}^d)^{\otimes n}$, the goal is to determine whether the state is fully product or $\varepsilon$-far from any such state in trace distance~\cite{montanaro2016Survey}. We will refer to this variation as \textit{multipartite product} (MP) testing (see Def.~\ref{def:BP-MP}). They construct an efficient product testing algorithm that utilizes multi-copy measurements and at most $T=O(1/\varepsilon^2)$ copies of the quantum state. Harrow and Montanaro then used this tester to prove that the complexity class QMA(2)=QMA(k) for $k>2$, a celebrated result in quantum complexity theory~\cite{harrow_testing_2010,Nishimura2025SwapSurvey}. We also note that their analysis has since been improved in Ref.~\cite{soleimanifar2022testingmatrixproductstates} and extended to the ``tolerant'' testing regime in Ref.~\cite{bakshi2024Learning}.

Another variant of product testing was first studied in Ref.~\cite{harrow2017Sequential}. One of the main results of this work is an algorithm that, given $T$ copies of a pure $n$-qudit state, can determine whether there exists some cut across which the state is product or whether it is $\varepsilon$-far from any such state in trace distance. We will refer to this task as \textit{bipartite product} (BP) testing. Their algorithm utilizes a sequence of multi-copy measurements requiring only $O(n/\eps^2)$ copies of the state, but with exponential time complexity. In the multi-copy setting, this is sample-optimal up to logarithmic factors in light of the $\Omega(n/\log{n})$ lower bound from Ref.~\cite{jones2024Testing}. The authors of Ref.~\cite{bouland2024state} recently provided a time- and sample-efficient algorithm for this task before generalizing it to find the specific cut across which the state is product, without sacrificing sample complexity. They referred to this as \textit{locating unentanglement}. Our main results, to which we now turn, have implications for the sample complexities of all of the aforementioned tasks when one is restricted to single-copy measurements. 

While we will formally define many of these notions in Sec.~\ref{sec:prelims}, we recall here the definition of an $\varepsilon$-tester in the realm of quantum property testing (see also the review \cite{montanaro2016Survey}) to state our main results. 

\begin{definition}[$\varepsilon$-tester] \label{def:eps-tester}
    Let $\PC$ denote the set of all quantum states $\ket{\psi} \in (\mathbb{C}^d)^{\otimes n}$ satisfying some property. Given $T$ copies of $\ket{\psi}$ and the promise that $\ket{\psi} \in \PC$ or $\ket{\psi}$ is $\varepsilon$-far from any element of $\PC$ in trace distance, an $\varepsilon$-tester is a two-outcome $POVM$ $\{M,\mathbb{I}-M\}$ with the following properties
    \begin{itemize}
        \item \textbf{Completeness.} If $\ket{\psi} \in \PC$, then $\tr{M \psi^{\otimes T}} \geq \frac{2}{3} \eqqcolon c. $
        \item \textbf{Soundness.} If $\ket{\psi}$ is $\varepsilon$-far (in trace distance) from any state in $\PC$, then $\tr{M \psi^{\otimes T}} \leq \frac{1}{3} \eqqcolon s$.
    \end{itemize}
    We often refer to a tester's bias, which is defined as $b\coloneqq c-s$. By our definition, our bias will be $1/3$, however, we note that the choice of this constant is arbitrary.
    \end{definition}
 In general, the assumption that a tester exists should be taken to mean that there is a tester such that $\Omega(1) \leq b$ (i.e., the bias is at least constant). To prove lower bounds on the sample complexity of a tester, then, it suffices to show an upper bound on the bias that decays to zero unless one takes $T$ sufficiently large. The result is a sample complexity lower bound of the form $T = \Omega(f(n,d,\eps))$. Let us outline this strategy at a high level before stating our main results. In classical and quantum learning theory, sample complexity lower bounds on property testing algorithms are often derived in the following manner.
\begin{enumerate}
    \item \textbf{Existence Assumption.} Assume that an algorithm exists which can solve the testing problem at hand. By this assumption, we can conclude that the probability of success $p_{\rm succ}$ must be at least a constant. That is, we assume $\Omega(1)\leq p_{\rm succ}$.
    \item \textbf{Ensemble Construction.} 
    Construct two ensembles of random states $\EC_1, \EC_2$ such that 
    \begin{align}
        \Pr_{\psi \sim \EC_1}[\psi \text{ is } \varepsilon\text{-close to}~\PC] = 1-o(1) \quad \text{and} \quad \Pr_{\psi \sim \EC_2}[\psi \text{ is } \varepsilon\text{-far from}~\PC] = 1-o(1).
    \end{align}
    In other words, the first ensemble contains states that the tester will \emph{accept} with high probability and the second ensemble should contain states that the tester will \emph{reject} with high probability.
    \item \textbf{Proving Hardness.} To prove distinguishing $\EC_1$ from $\EC_2$ with single-copy measurements is difficult, one must prove that the probability distributions over classical measurement outcomes resulting from these ensembles are hard to distinguish (i.e., they are close in statistical distance). This will imply an upper bound on $p_{\rm succ}$ that decays to zero unless the algorithm uses enough samples. Because we assumed that a tester, should it exist, would have a constant lower bound on the probability of success, we can solve for the  condition on the number of samples that ensures these inequalities hold, yielding a lower bound on the sample complexity of any tester for $\PC$.
\end{enumerate}
The non-trivial part of this process, then, is to construct two such ensembles which are provably difficult to distinguish using only single-copy measurements. In both variants of product testing, we considered distinguishing ensembles of global Haar random pure states (which are far from product with high probability, see Proposition~\ref{prop:far-from-BP}) from ensembles of product states where each local state was itself Haar random.

\subsection{Results}
Our first main result concerns BP testing with single-copy measurements. That is, given $T$ copies of $\ket{\psi} \in (\mathbb{C}^d)^{\otimes n}$, we want to determine whether there exists a cut $S \subset [n]$ for which $\ket{\psi}$ is product across $S:S^c$ or if $\ket{\psi}$ is far from any such state. Put another way, we are asked to determine whether $\ket{\psi}$ is genuinely multipartite entangled (GME). This problem is motivated by several applications for which GME states are a critical resource including: quantum communication and cryptography, measurement-based quantum computation, quantum networks, foundations of quantum mechanics, and error correction~\cite{horodecki2009Quantum}.

\begin{theorem}[Single-copy Lower Bound on BP Testing, Thm.~\ref{thm:BP-test-LB} Restated] Any algorithm using, potentially adaptive, single-copy measurements to test whether a state $\ket{\psi} \in (\mathbb{C}^d)^{\otimes n}$ is product across some cut, or is $\varepsilon$-far from any such state (with probability at least $2/3$) must use at least $\Omega(d^{n/4})$ samples. 
\end{theorem}

An algorithm utilizing a sequence of two-outcome measurements was shown in Ref.~\cite{harrow2017Sequential} to require only $O(n/\eps^2)$ copies of the state. Thus, our single-copy lower bound proves an exponential separation between single- and multi-copy measurements for this task. Ref.~\cite{harrow2017Sequential}'s algorithm leveraged the two-copy product test, which accepts with certainty when the state is product across a given cut and with probability $1-\Theta(\varepsilon^2)$ as long as the state is $\varepsilon$-far from product. Using $O(n/\varepsilon^2)$ copies of the state, the acceptance probability in the non-product case can be made as small as $2^{-\Omega(n)}$. The quantum OR bound allows the product tests for each of the $2^{n}-2$ bipartitions to be combined to test for genuine multipartite entanglement by making use of the gentle measurement lemma~\cite{gao2015Quantum,oskouei2019Union,odonnell2021Quantum}. In the multicopy setting, this is sample-optimal up to logarithmic factors in light of the $\Omega(n/\log{n})$ lower bound from Ref.~\cite{jones2024Testing}.

Unfortunately, the above algorithm had exponential time complexity, a shortcoming that was removed recently in Ref.~\cite{bouland2024state}. Here, the authors 1) made the algorithm time-efficient and 2) generalized it to find the specific cut across which the state is product, without sacrificing sample complexity. Therefore, our lower bound on BP testing immediately implies a lower bound on any algorithm using only single-copy measurements that could be used for this so-called \textit{hidden cut problem}. As such, any single-copy algorithm must use at least $\Omega(d^{n/4})$ copies while there exists an algorithm utilizing multi-copy measurements that needs only $O(n/\varepsilon^2)$ copies.

 Our next result concerns a sample lower bound for any single-copy MP testing algorithm. As mentioned above, in Ref.~\cite{harrow_testing_2010}, the authors provide a simple two-copy algorithm for MP testing using $\Theta(1/\eps^2)$ samples (the analysis was later simplified and improved in Ref.~\cite{soleimanifar2022testingmatrixproductstates}). In a short appendix, they also showed that any product tester on two copies of an $n$-qudit state that utilizes measurement operators that have a positive-partial transpose (PPT)\footnote{Note that PPT protocols are a relaxation of the class of all protocols  using only local operations and classical communication (LOCC)~\cite{chitambar2014Everything}.} must have a bias that scales as $O(1/d)$, implying a sample complexity that must scale with the local dimension of the state. To show that this is still the case for such testers given more than just two copies of the state, significant mathematical machinery is needed. In Ref.~\cite{harrow_testing_2010}, a forward citation is given to a paper that finally appeared in 2023~\cite{harrow2023Approximate}. In this work, Harrow uses the approximate orthogonality of permutation operators to derive several interesting results, including a $\Omega\left(d^{1/8}\right)$ lower bound on the $n=2, T>2$ case of product testing with PPT measurements.\footnote{Note also that our choice of parameter notation differs from Ref.~\cite{harrow2023Approximate} in which the author considers access to $n$ copies of a $k$ qudit state. Here, we consider access to $T$ copies of an $n$-qudit state.}  We sharpen this result in the $n=2$ case by considering a more restrictive (but experimentally motivated) notion of single-copy measurements, and we extend the regime of validity to all $n \geq 2$. 

\begin{theorem}[Single-copy Lower Bound on MP Testing, Thm.~\ref{thm:MP-test-LB} Restated] Any algorithm using, potentially adaptive, single-copy measurements to test whether a state is product across all cuts, or is $\varepsilon$-far from any such state (with probability at least $2/3$) must use at least $\Omega(\sqrt{d/n})$ samples.
\end{theorem}

Of course, this bound is most useful when $d\gg n$. While larger lower bounds could likely be found by considering more difficult distinguishing tasks, for $n=2$, we have an algorithm with sample complexity matching for constant $\varepsilon \leq 1/2$, implying the sample complexity of single-copy product testing in this regime is $\Theta(\sqrt{d})$.
\begin{theorem}[Single-copy, Local Upper Bound on MP Testing, Informal Version of Thm.~\ref{thm:MP-test-UB}] There exists a non-adaptive algorithm requiring only single-copy, local measurements on $O(n \log{n} \sqrt{d}/\eps^2)$ samples of a pure state to test whether it is fully product, or is $\varepsilon$-far from any such state with probability at least $2/3$. 
\end{theorem}
The underlying idea of the algorithm is quite simple: if an $n$-qudit pure state $\ket{\psi}\in(\mathbb{C}^d)^{\otimes n}$ is $\varepsilon$-far in trace distance from the set of all MP states, then one would expect that at least one of the reduced qudit states is sufficiently impure, i.e., $1- \eps_{\mathrm{rej}} \geq \tr{\psi_i^2}$ holds for at least one $i$ and some tolerance $\eps_\mathrm{rej}$. The key technical challenge is to show that, for a state that is $\varepsilon$-far from an MP state, at least one marginal is sufficiently impure. Then, we can utilize a recent algorithms for single-copy purity estimation~\cite{anshu2022Distributed,gong2024sample} to estimate the purity of each local $\psi_i$ to sufficient accuracy to solve the testing problem with high probability. 

We note that, after posting an initial draft to arXiv, we were made aware of Ref.~\cite{liu2025Separation}, which proves exponential separations for various entanglement testing tasks. In particular, Theorem 2 in Ref.~\cite{liu2025Separation} tightens the sample complexity lower bound for product testing (with $n=2$ parties, i.e. when MP and BP are the same task) present in Theorem 9 of Ref.~\cite{harrow2023Approximate} from $\Omega(d^{1/8})$ to $\Omega(d^{1/4})$. Unaware of their work, we independently proved our Lemma 3.3 in order to then obtain our exponential lower bound on BP testing. Their Lemma C1 is an equivalent result proved using tensor network diagrams complementing our group theoretic, inductive approach. 

\section{Preliminaries}\label{sec:prelims}
Let us begin by introducing some terminology and notation. The problems we consider have two underlying tensor structures: one within the multi-qudit state and one between multiple copies of such a state. When both are present in an expression, we will denote the space of $T$-fold tensor product of a $n$-qudit state as $((\mathbb{C}^d)^{\tilde{\otimes}n})^{\otimes T}$, to avoid confusing the two tensor products. Recalling Def.~\ref{def:BP-MP}, we will denote the set of all BP states on $(\mathbb{C}^d)^{\otimes n}$ as $\BC_n$ and the set of all MP states as $\MC_n$. With this notation in place, let us collect some essential facts from quantum information theory.

\subsection{Quantum States, Measurements, and Distances}
\subsubsection{States and Distances}
Throughout this work, let $\ket{\psi} \in \mathbb{C}^d$ and denote the corresponding projector as $\psi \coloneqq \ketbra{\psi}{\psi}$. We will denote the maximally mixed state as $\rho_{mm}\coloneqq \frac{\mathbb{I}}{d}$, where the underlying dimension will be clear from context. Further, let the vector space of linear operators on a Hilbert space $\HC$ be denoted $\LC(\HC)$ and the set of all density operators $\DC(\HC) = \{X \in \LC(\HC) : X \geq 0, \tr{X}=1 \}$. Given a state $\phi\in \DC(\HC_1\otimes \cdots \otimes \HC_n)$ we will use the notation $\phi_i=\mathrm{tr}_{[n]\setminus \{i\}}[\phi]$ to denote the reduced state on site $i\in [n]\coloneqq \{1,\ldots,n\}$ or for $\phi\in \DC(\HC_A\otimes \HC_B)$ we may use the superscript $\phi^A=\mathrm{tr}_{A^c}[\phi]$ to denote the reduced state or a combination thereof. Additionally, given a subset $S\subset [n]$, we will use the notation $S^c\coloneqq[n]\setminus S$ to denote its compliment. With this notation in place, we may formally state the following variants of productness considered in this work (following the naming conventions of Ref.~\cite{jones2024Testing}).

% \anote{Throughout the text, we will be interested in states with a variety of different product structure:
% \begin{enumerate}
%     \item \Multipartite product states
% \end{enumerate}
% }

\begin{definition}[Multipartite and Bipartite Productness]\label{def:BP-MP}
    A pure $n$-qudit state $\ket{\psi} \in (\mathbb{C}^d)^{\otimes n}$ is 
    \begin{itemize}
        \item \textbf{Genuinely multipartite entangled} (GME) if is it is entangled across every bipartition of the $n$ parties.
        \item \textbf{Bipartite product} (BP) if there exists at least one non-trivial subset $S \subset [n]$ such that the state is product across the cut $S:S^c$ (i.e., if it is not GME). We will denote the set of all such $n$-qudit pure states $\BC_n$.
        \item \textbf{Multipartite product} (MP) if the state is of the form $\ket{\psi} = \bigotimes_{i=1}^n \ket{\phi_i}$ for arbitrary pure states $\ket{\phi_i}$. That is, if it is product across \textit{every} partition. We will denote the set of all such $n$-qudit pure states $\MC_n$.
    \end{itemize}
    Note that $\MC_2 = \BC_2$ but $\MC_n \subset \BC_n$ for all $n>2$.
\end{definition}

 \begin{definition}[Trace Distance]\label{def:trace-distance}
        Let $\rho$ and $\sigma$ be quantum states on $\HC$. The trace distance between them is 
        \begin{align}
            d_{\rm tr}(\rho,\sigma)\coloneqq\frac{1}{2}\|\rho -\sigma\|_1,
        \end{align}
        where $\|X\|_1\coloneqq\tr{\sqrt{X^{\dagger}X}}$ for all $X \in \LC(\HC)$.
    \end{definition}
    When the input states are pure (i.e. $\rho = \ketbra{\psi}{\psi}$ and $\sigma = \ketbra{\phi}{\phi}$), the trace distance takes the simple form $d_{\rm tr}(\ketbra{\psi}{\psi},\ketbra{\phi}{\phi}) = (1- \abs{\braket{\psi|\phi}}^2)^{1/2}$. Next, note that because the trace distance is a function of the spectrum of $\rho-\sigma$, it is unitarily invariant. This will allow us to define what it means to be $\varepsilon$-close to a subset of quantum states in a well-defined manner. 
    \begin{definition}[Distance from a Subset]\label{def:subset-distance}
        Let $\PC \subseteq \DC(\HC)$ be a subset of all quantum states. Then, we say a state $\psi \in \DC(\HC)$ is $\varepsilon$-far from $\PC$ if 
        \begin{align}
            \varepsilon \leq \min_{\varphi \in \PC} d_{\rm tr}(\varphi,\psi).
        \end{align}
    \end{definition}
Finally, because we are restricting our attention to pure states and often discussing bipartitions of such states, it will be very useful to consider the Schmidt decomposition of the state~\cite{nielsen2010Quantum}.
    \begin{definition}[Schmidt Decomposition]\label{def:schmidt-decomp}
        A bipartite pure state always admits a Schmidt decomposition 
    \begin{align}
        \ket{\psi}_{AB} &= \sum_{i=1}^r \sqrt{\lambda_i} \ket{i_A}\otimes \ket{i_B},
    \end{align}
    where the $r$ non-zero Schmidt coefficients $\sqrt{\lambda_i}$ characterize the entanglement in the state. Moreover, the states $\{\ket{i_A}\}$ and $\{\ket{i_B}\}$ form orthonormal sets.
    \end{definition}
For a bipartite pure state to be a product state, the Schmidt rank (the number of nonzero $
\lambda_i$'s) must be equal to one. We will utilize this decomposition in our constructions below.

\subsubsection{Measurement Classes}
As classical creatures living in a quantum world, the only way to obtain information about our quantum state is to measure it. The most general measurements allowed by quantum mechanics are formally defined as follows. 

\begin{definition}[Positive Operator-valued Measure (POVM)]\label{def:POVM}
A positive operator-valued measure (POVM) is a collection of operators $\{E_k\}_k$ on a Hilbert space $\HC$ that satisfy:
\begin{enumerate}
    \item Positivity: $E_k \geq 0$ for all $k$
    \item Completeness: $\sum_k E_k = \mathbb{I}$.
\end{enumerate}
When one applies a POVM to a state $\rho$, the probability of obtaining outcome ``$k$'' is given as $p_k = \tr{\rho E_k}$.
\end{definition}

In this work, we will be considering testing in the single-copy setting. Although multi-copy strategies typically have lower sample complexity, the required measurements remain largely infeasible on near-term devices. As such, it is of great practical interest to determine the separation between these single- and multi-copy strategies. Moreover, among single-copy measurements, even enacting a global measurement across tens of systems can be very challenging, which motivates studying separations between single-copy global and single-copy local measurements. To ease this discussion, we define the following classes of measurement strategies.

\begin{definition}[Locality of POVMs]\label{def:locality-of-POVMs}
Let $\mathcal{H} = \bigl( (\mathbb{C}^d)^{\tilde{\otimes} n} \bigr)^{\otimes T}$ denote the Hilbert space of $T$ copies of an $n$-qudit system.  
A POVM on $\mathcal{H}$ is classified as follows\footnote{In the literature, many works refer to multi-copy measurements as entangled or coherent, implying single-copy measurements are unentangled or incoherent. We avoid this terminology to due to the ambiguity regarding entanglement between $T$ copies of $\ket{\psi} \in (\mathbb{C}^{d})^{\tilde{\otimes} n}$ and entanglement within a single copy of $\ket{\psi} \in (\mathbb{C}^{d})^{\tilde{\otimes} n}$.}:
\begin{itemize}
    \item \textbf{Multi-copy:} the POVM may act jointly across all $T$ copies, i.e.\ on the full space $\mathcal{H}$.
    \item \textbf{Single-copy, global:} the POVM acts on one $n$-qudit system at a time, i.e.\ on $(\mathbb{C}^d)^{\tilde{\otimes} n}$.
    \item \textbf{Single-copy, local:} the POVM factors as a tensor product across qudits and acts individually on each $\mathbb{C}^d$.
\end{itemize}
\end{definition}

Of particular importance, will be the case when the $E_k$'s are rank-1, in which we call the collection of operators a \textit{rank-1 POVM}. In this case, we may associate with each $E_k$ a pure quantum state $\ket{\psi_k}$. Formally, a rank-1 POVM is a (possibly overcomplete) set of rank-1 operators $\{d\cdot a_k \ketbra{\psi_k}{\psi_k}\}_k,$ with $a_k \geq 0$ such that 
\begin{align}
    \sum_k d \cdot a_k  \ketbra{\psi_k}{\psi_k} = \mathbb{I}_{\HC}
\end{align}
where the factor of $d$ ensures the $\sum_k a_k =1$. Concretely, if we measure $\rho$, the probability of obtaining the outcome ``$k$'' is given as $p_k=\tr{d \cdot a_k \cdot \rho \ketbra{\psi_k}{\psi_k}}$ and the corresponding post-measurement state is simply $\ketbra{\psi_k}{\psi_k}$. Rank-1 POVMs are particularly important when only the classical outcome is relevant (i.e., when the post-measurement state is discarded). In this case, all POVMs can be simulated by a rank-1 POVM with post-processing. 

\begin{lemma}[Simulating POVMs with Rank-$1$ POVMs] \label{lem:simulating-POVMs}
    Assuming that the post-measurement state can be neglected, an arbitrary POVM on $\mathbb{C}^d$ with $N$ outcomes can be simulated by a rank-$1$ POVM with $M \leq N \cdot d$ outcomes.
\end{lemma}
   The idea is simply that every rank-$k$ POVM element is Hermitian and positive, thus diagonalizable. So, to simulate a rank-$k$ POVM with a rank-$1$ POVM, we just diagonalize the rank-$k$ elements (yielding a weighted sum of rank-$1$ operators). Then, we implement a POVM with all of these rank-$1$ elements and group the outcomes accordingly. See Proposition~2 in Ref.~\cite{haapasalo2012Quantum} or Lemma~4.2 in Ref.~\cite{chen2022Exponential} for a formal proof of this fact. With these fundamentals in place, we now collect some essential facts from representation theory.

\subsection{Representation Theory Essentials} \label{sec:rep-theory-essentials}
By exploiting symmetries in physical systems, one can often greatly reduce the complexity of a problem. The tools that allow us to formally study such symmetries come from group and representation theory. They have become essential tools for many sub-fields modern physics and mathematics, quantum information theory included. Below, we provide only the essentials needed to understand our main results. We refer the reader to Ref.~\cite{harrow2013Church} and Ref.~\cite{mele2024Introduction} for a more comprehensive treatment of the results of this section.

Let $\SC_n$ denote the \textit{symmetric group} on $n$ objects, that is, the group of bijections on a set with a cardinality of size $n$. Below we give an important lemma regarding permutations of the symmetric group which will be used later in this work.

\begin{lemma}[Double Coset Decomposition] \label{lemma:double-coset-decomp}
Consider the symmetric group on $T+1$ objects, denoted $\SC_{T+1}$ and let \( \SC_T^{[2,T+1]} \subset \SC_{T+1} \) denote the subgroup of permutations that act only on indices \(2\) through \(T+1\). Then every permutation \( \pi \in S_{T+1} \) can be written as
\begin{align}
\pi = \alpha \cdot (1\ 2)^a \cdot \beta,
\end{align}
where \( a \in \{0,1\} \) and \( \alpha, \beta \in \SC_T^{[2,T+1]} \).
\end{lemma}

\begin{proof}
Consider the action of \( \SC_{T+1} \) on index \(1\). Any permutation \( \pi \in \SC_{T+1} \) either fixes \(1\), in which case
\begin{align}
\pi \in \SC_T^{[2,T+1]},
\end{align}
or moves \(1\) to some \( j \in \{2, \dots, T+1\} \). In the latter case, choose \( \gamma \in \SC_T^{[2,T+1]} \) such that $\gamma(j) = 2.$ Then
\begin{align}
\gamma \circ \pi(1) &= \gamma(j) = 2 \implies 
(1\ 2) \circ \gamma \circ \pi(1)= 1.
\end{align}
So the permutation $\delta \coloneqq (1\ 2) \circ \gamma \circ \pi$ fixes index 1, hence $
\delta \in \SC_T^{[2,T+1]}.$ Solving for \( \pi \), we get
\begin{align}
\pi = \gamma^{-1} \cdot (1\ 2) \cdot \delta.
\end{align}
This gives the desired form with \( \alpha = \gamma^{-1} \), \( a = 1 \), and \( \beta = \delta \). If \( \pi \) already fixes index \(1\), then the decomposition holds trivially with
\begin{align}
a = 0, \quad \alpha = \pi, \quad \beta = e,
\end{align}
which concludes the proof.
\end{proof}
Permutations arise naturally when one considers their action on tensor products of quantum systems. Specifically, for $\pi \in \SC_n$, define 
\begin{align}\label{eq:action-of-sym-group}
    P_d(\pi) \ket{i_1} \otimes \dotsm \otimes \ket{i_n} &= \ket{i_{\pi^{-1}(1)}} \otimes \dotsm \otimes \ket{i_{\pi^{-1}(n)}},
\end{align}
to be the action of the symmetric group on $(\mathbb{C}^{d})^{\otimes n}$. From this definition, one can show that $P_d(e) = \mathbb{I}$ and for all $\pi,\sigma \in \SC_n$, we have $P_d(\pi\cdot \sigma)= P_d(\pi)P_d(\sigma)$. Together, these properties make $P_d(\pi)$ a \textit{unitary representation} of $\pi \in \SC_n$ on $(\mathbb{C}^{d})^{\otimes n}$. In the proofs below, when the local dimension is clear from context, we will simply let $\pi$ represent both the permutation and the unitary representation. 

In the case of $\SC_2$ we define $P_d((1\ 2))\eqqcolon\mathbb{F}$ to be the \textit{swap} operator. This operator has the following important property which will be utilized throughout the text.

\begin{lemma}[The Swap ``trick''] \label{lemma:swap-trick}
Given two linear operators on a finite-dimensional dimensional Hilbert space, $A$ and $B$, the following equality holds
\begin{align}
    \tr{\mathbb{F} \cdot A \otimes B} &= \tr{AB}.
\end{align}
This is commonly referred to as the swap ``trick.''
\end{lemma}
The proof of the swap trick is standard and can be verified by explicit computation. Next, we will also need several facts about the symmetric subspace~\cite{harrow2013Church,mele2024Introduction}, the subset of all quantum states that are permutationally invariant. The projector onto this subsapce is defined as
\begin{align}
    \Pi_{\rm sym}^{d,T} &= \frac{1}{T!} \sum_{\pi \in S_T} P_d(\pi),
\end{align}
with the dimension of the symmetric subspace given as
\begin{align}
   \tr{\Pi_{\rm sym}^{d,T}} &= \binom{d+T-1}{T}. 
\end{align}
From these basic definitions, one can show that for any $\pi \in S_T$, $P_d(\pi)\Pi_{\rm sym}^{d,T} = \Pi_{\rm sym}^{d,T} P_d(\pi) = \Pi_{\rm sym}^{d,T}$. 

Next, let us denote the unitary group $\UC_d \coloneqq \{U \in \LC(\mathbb{C}^d) : U^{\dagger}U = \mathbb{I}_d\}$. The unitary group admits a unique uniform measure called the \textit{Haar measure} which will allow us to take uniform averages over $\UC_d$. For a great review of the Haar measure with quantum information theorists in mind, see Ref.~\cite{mele2024Introduction}.

\begin{definition}[Haar Measure on $\UC_d$]\label{def:haar}
    The Haar measure on the unitary group $\UC_d$ is the unique probability measure $dU$ that is left- and right-invariant over $\UC_d$. That is, for all integrable functions $f$ and all $V \in \UC_d$, we have
    \begin{align}
        \int_{\UC_d} f(U) dU =  \int_{\UC_d} f(VU) dU =  \int_{\UC_d} f(UV) dU.
    \end{align}
\end{definition}
We note that for any (measurable) $S \subseteq \UC_d$, it follows that $\int_S 1 dU \geq 0$ and $\int_{\UC_d} 1 dU = 1$, which makes the Haar measure a probability measure. Of particular importance in this work, will be Haar random pure states, which can be obtained by simply applying a Haar random unitary to an arbitrary pure state.
\begin{definition}[Haar Random Pure State] \label{def:Haar-random-state}
Let $\mathbb{C}^d$ be a $d$-dimensional Hilbert space.  
A Haar random pure state on $\mathbb{C}^d$ is defined as
\[
\ket{\psi} = U \ket{\phi},
\]
where $\ket{\phi} \in \mathbb{C}^d$ is any fixed pure state and 
$U \in \UC_d$ is drawn from the unitary group according to the Haar measure.
\end{definition}

Taking a Haar average over tensor powers of such states will yield a state that is unitarily invariant (by the invariance of the Haar measure) and also permutationally invariant. The latter fact can be seen from the following result which we will utilize throughout.

\begin{lemma}[Average State Symmetric Subspace Projector~\cite{harrow2013Church}]\label{lem:sym-sub-proj} \index{Symmetric Subspace Projector}
    Let $\ket{\psi} \in \mathbb{C}^d$ and define $\psi \coloneqq \ketbra{\psi}{\psi}$. Then, we have
    \begin{align}
        \mathbb{E}_{\psi}[{\psi^{\otimes T}}]&=\frac{\Pi_{\rm sym}^{d,T}}{\dim{\Pi_{\rm sym}^{d,T}}}= \frac{\sum_{\pi \in S_T}P_d(\pi)}{d(d+1)\dotsm (d+T-1)},
    \end{align}
    where $\mathbb{E}_{\psi}[f(\psi)]$ denotes the average over Haar random pure states in $\mathbb{C}^d$.
\end{lemma}
See Proposition 6 in Ref.~\cite{harrow2013Church} or Theorem 22 in Ref.~\cite{mele2024Introduction} for a proof of this result. With these quantum results in mind, we now turn to a brief review of some key topics in classical statistical learning theory that are especially relevant to the analysis of single-copy testing algorithms.

\subsection{Tools from Statistical Learning Theory}
In single-copy measurement protocols, we are given one copy of $\rho$ at a time and may measure in a basis we choose and then must store the outcome in classical memory before receiving a fresh copy of $\rho$. In this single-copy regime, we must work with classical distributions over measurement outcomes. Thus, it will be helpful to review some essential results from the vast body of statistical learning theory literature.   
\subsubsection{Distinguishing Probability Distributions}
\begin{definition}[Total Variation Distance]
    Let $\mu_p$ and $\mu_q$ be two probability measures over a domain $\Omega$. Then, the total-variation distance between $\mu_p$ and $\mu_q$, is defined as 
\begin{align}\label{eq:TV-distance-general}
    d_{\mathrm{TV}} (\mu_p,\mu_q) = \sup_{S \subseteq \Omega} |\mu_p(S)-\mu_q(S)|
\end{align}
where the sup is taken over well-behaved sets (i.e., measurable).
\end{definition}
Intuitively, the total variation (TV) distance captures the largest difference that two probability measures can take on any event. We will be working in countable probability spaces, within which the TV distance takes the simple form
\begin{align}
        d_{\mathrm{TV}}(p,q) = \frac{1}{2} \sum_{\omega \in \Omega} |p(\omega) -q(\omega)|,
\end{align} 
where $p,q$ denote the probability densities of the probability measures $\mu_p,\mu_q$. The TV distance is particularly important due to its operational interpretation as the maximum bias one can achieve when attempting to distinguish two distributions $p,q$. That is, the probability of successfully distinguishing $p$ and $q$ is upper bounded as
\begin{align}\label{eq:succ-prob}
    p_{\rm succ} \leq \frac{1}{2}+\frac{1}{2}d_{\rm TV}(p,q),
\end{align}
when one assumes a uniform prior. In addition to its important operational interpretation, the TV distance also forms a metric. Thus, it is symmetric, non-negative, vanishes only when $p=q$, and satisfies the triangle inequality 
\begin{align}
    d_{\rm TV}(p,q) \leq d_{\rm TV}(p,r) + d_{\rm TV}(r,q),
\end{align}
for any three distributions $p,q,r$ on the same probability space. For quantum information theorists, we note that this is simply the classical analogue of trace distance. A standard method of deriving lower bounds on the sample complexity of distinguishing classical distributions is to construct two ensembles over distributions that are hard to distinguish unless one uses sufficiently many samples. By assuming the existence of a testing or learning algorithm, we are assuming that $\Omega(1) \leq p_{\rm succ}$ (i.e., the success probability is at least constant). Such an assumption, along with Eq.~\eqref{eq:succ-prob}, implies that 
\begin{align}
\Omega(1) \leq d_{\rm TV}(p,q).
\end{align}
Thus, if we can prove an upper bound on the TV distance that decays to zero unless the number of samples is sufficiently large, we will obtain the desired sample complexity lower bound. As Ref.~\cite{chen2022Exponential} points out, directly proving an upper bound on the TV distance is essentially proving a two-sided bound which can be quite challenging. Fortunately, the following lemma will allow us to focus on the easier task of proving a one-sided bound. 

\begin{lemma}[One-sided Likelihood Ratio Bounds]\label{lem:one-sided-ratios} Let $0\leq \delta < 1$. Let $p,q$ be two probability distributions on a finite domain $\Omega$. Then, 
\begin{align}
\frac{p(\omega)}{q(\omega)} \geq 1- \delta \quad \forall~\omega \in \Omega \implies d_{\rm TV}(p,q) \leq \delta.
\end{align}
\end{lemma}
\begin{proof}
    First, let us recall that TV distance is a metric and therefore symmetric, 
    \begin{align}
        d_{\rm TV}(p,q) &= \frac{1}{2} \sum_{\omega \in \Omega}|p(\omega)-q(\omega)| = \frac{1}{2} \sum_{\omega \in \Omega}|q(\omega)-p(\omega)|.
    \end{align}
Then, by breaking up the sum into the cases where $p(\omega) \geq q(\omega)$ (and vice-versa), one can show
\begin{align}
    \frac{1}{2} \sum_{\omega \in \Omega} \left|q(\omega) - p(\omega)\right| &= \sum_{\omega: q(\omega) \geq p(\omega)} q(\omega) - p(\omega).
\end{align}
Using this fact, the symmetry of TV distance, and the assumption $p(\omega) \geq (1-\delta)q(\omega)$, we obtain
\begin{align}
   d_{\rm TV}(p,q)  & = \frac{1}{2} \sum_{\omega \in \Omega}|q(\omega)-p(\omega)| ,\\
    &=\sum_{\omega: q(\omega) \geq p(\omega)} q(\omega) - p(\omega),\\
    &\leq \sum_{\omega: q(\omega) \geq p(\omega)} q(\omega) - (1-\delta)q(\omega),\\
    &\leq \sum_{\omega: q(\omega) \geq p(\omega)} \delta q(\omega),\\
     d_{\rm TV}(p,q) &\leq \delta,
\end{align}
as desired. 
\end{proof}
We now focus our attention on the form of the probability distributions over measurement outcomes that will result from a single-copy learning/testing algorithm.
\subsubsection{Learning with Restricted Measurements}
The setting we consider is motivated by near-term experiments that cannot store many identical copies of a quantum state at once. We imagine there is some experiment that can reliably produce copies of a state $\rho$. Due to our limited resources, we can only store one copy of $\rho$ at a time, so we must apply some single-copy POVM $\{E_k\}_k$ to obtain a classical outcome $k$ with probability $p_k = \tr{\rho E_k}$ before discarding the post-measurement state and requesting a fresh copy of the state. After $T$ rounds of this protocol, we must make our prediction of some property of $\rho$. We say the that $T$ is the \textit{sample complexity} of this learning or testing algorithm. Naturally, the goal is to minimize $T$ subject to various constraints.

The data we will use to make our prediction will be some classical string of outcomes generated through the experiment. Given an underlying state $\rho$, the probability of obtaining outcome $s_t$ in the $t$-th round of the experiment is $p^{\rho}(s_t)=\tr{d\cdot a_{s_t}\cdot \rho \ketbra{\psi_{s_t}}{\psi_{s_t}}}$, where we have used Lemma~\ref{lem:simulating-POVMs} to restrict our attention to rank-1 POVMs. By assumption, we are given an identical copy of $\rho$ in each successive round, so the probability of obtaining classical outcome string $\ell = s_1 \dotsm s_T$ is given by the product distribution
\begin{align}\label{eq:classical-outcome-distribution}
    p^{\rho}(\ell) &= \prod_{t=1}^T \tr{d\cdot a_{s_t}\cdot \rho \ketbra{\psi_{s_t}}{\psi_{s_t}}}.
\end{align}
If we are promised the underlying state is either $\rho$ or $\sigma$, our challenge becomes distinguishing $p^{\rho}(\ell)$ and $p^{\sigma}(\ell)$. Thus far, we have described the simplest setting in statistical learning theory: distinguishing two distributions (i.e., one-to-one distinguishing). However, one can often prove stronger lower bounds by considering many-to-one or many-to-many distinguishing tasks. In a more general scenario, we are tasked with distinguishing between two collections of states $\EC_1 = \{\rho_x\}_{x \sim \XC}$ and  $\EC_2=\{\rho_y\}_{y \sim \YC}$ which are sampled with respect to some specified distributions $\XC, \YC$. In this more general setting, an upper bound on the probability of distinguishing between $\EC_1$ and $\EC_2$ is given by \textit{Le-Cam's two-point method}~\cite{Yu1997}.
    \begin{lemma}[Le Cam's Two-point Method] \label{lem:LeCam-2pt}
    Let $\mathscr{L}(T)$ denote the set of all possible classical measurement outcome strings after $T$ iterations of a classical learning algorithm. Then the maximum probability $p_{\rm succ}$ of distinguishing $\EC_1 = \{\rho_x\}_{x \sim \XC}$ and $\EC_2=\{\rho_y\}_{y \sim \YC}$, obeys
\begin{align}
    p_{\rm succ} \leq \frac{1}{2}+ \frac{1}{4} \sum_{\ell \in \mathscr{L}(T)} |\mathbb{E}_x p^{\rho_x} (\ell) -  \mathbb{E}_y p^{\rho_y} (\ell)| \eqqcolon \frac{1}{2} + \frac{1}{2}d_{\rm TV}(\mathbb{E}_x p^{\rho_x} (\ell),\mathbb{E}_y p^{\rho_y} (\ell)).
\end{align}
\end{lemma}
With the absolute values, this quantity can be challenging to upper bound. To leverage results from the literature, we utilize the fact that the total variation distance is a metric and thus is symmetric and satisfies triangle inequality. We can then derive an upper bound on the desired total variation distance as 
    \begin{align}
        d_{\rm TV}(\mathbb{E}_x p^{\rho_x} (\ell),\mathbb{E}_y p^{\rho_y} (\ell)) \leq d_{\rm TV}(\mathbb{E}_x p^{\rho_x} (\ell),p^{\rho_{mm}}(\ell))+d_{\rm TV}(\mathbb{E}_y p^{\rho_y} (\ell),p^{\rho_{mm}}(\ell)),
    \end{align}
    where we have chosen the common reference state to be the maximally mixed state $\rho_{mm}$ because it will greatly simplify the analysis. It is then much simpler to derive one-sided bounds in the following manner.

    \begin{lemma}[One-sided Bound Suffices for Le Cam, Lem.~5.4~\cite{chen2022Exponential}]\label{lem:one-sided-LeCam} 
     Let $\mathscr{L}(T)$ denote the set of all possible classical measurement outcome strings after $T$ iterations of a classical learning algorithm. If for all $\ell\in \mathscr{L}(T)$, we have
    \begin{align}
        \frac{\mathbb{E}_x [p^{\rho_x} (\ell)]}{p^{\rho_{mm}}(\ell)} \geq 1-\delta, \quad 
    \end{align}
    then it follows that
    \begin{align}
        d_{TV}(\mathbb{E}_x p^{\rho_x}(\ell),p^{\rho_{mm}}(\ell))\leq \delta.
    \end{align}
\end{lemma}
If we then find $\delta_1$ such that $d_{TV}(\mathbb{E}_x p^{\rho_x}(\ell),p^{\rho_{mm}}(\ell)) \leq \delta_1$ and $\delta_2$ such that $d_{TV}(\mathbb{E}_y p^{\rho_y}(\ell),p^{\rho_{mm}}(\ell)) \leq \delta_2$, then our overall bound on the success probability in this many-to-many distinguishing task will be 
\begin{align}
    p_{\rm succ} \leq \max{\{\delta_1,\delta_2\}}.
\end{align}
Combining this upper bound with the assumption that a testing/learning algorithm exists yields 
\begin{align}
    \Omega(1) \leq \max{\{\delta_1,\delta_2\}},
\end{align}
where the dependence on $n,T$, and other parameters is suppressed in $\delta$. The goal is to prove that this upper bound decays to zero (contradicting the existence of a tester) unless sufficiently many samples is taken. To derive strong upper bounds on the total variation distances, we must prove strong lower bounds on the one-sided likelihood ratios in Lemma~\ref{lem:one-sided-LeCam}. Our main technical contribution is relating these likelihood ratios to permanents of Gram matrices and then utilizing the vast mathematical literature on this topic to lower bound the desired likelihood ratios. As such, we now review a few facts about permanents of PSD matrices. 
\subsection{Lower Bounds on Permanents of Positive Semi-Definite Matrices}
In this work, we will prove lower bounds on certain likelihood ratios that arise in the analysis of single-copy learning algorithms. In several instances, we will relate these quantities to permanents of Gram matrices, so we now recall the definition of the permanent.
\begin{definition}[Permanent]\label{def:permanent}
    The permanent of an $n$-by-$n$ matrix $A\coloneqq  (A_{ij})$ is defined as 
\begin{align}
    \mathrm{per}(A) = \sum_{\pi \in \SC_n} \prod_{i=1}^n A_{i\pi(i)}. 
\end{align}
\end{definition}
With this definition in place, we can state a useful lower bound on the permanent of PSD matrices with unit diagonal.
\begin{lemma}[Permanent of Correlation Matrix Lower Bound, Ref.~\cite{grone1988Permanental}]\label{lem:perm-cor-LB}
Let $A$ be an $n$-by-$n$ positive semi-definite matrix with $|a_{ii}|=1$ for all $1\le i \le n$. Then,
\begin{align}
    \mathrm{per}(A)\ge \frac{1}{n}\norm{A}_{F}^2
\end{align}
where $\norm{A}_F\coloneqq \sqrt{\tr{A^{\dagger}A}}$ is the Frobenius norm of A.
\end{lemma}
In this work, we will want to specifically apply this lemma to the permanent of a Gram matrix. We state this result in the following corollary.
\begin{corollary}[Permanent of Gram Matrix Lower Bound] \label{cor:perm-gram-LB}
Let $\{\ket{\psi_t}\}_{t=1}^T$ be a collection of $T$ pure qudit states and define $G_{tt'} \coloneqq  \braket{\psi_t|\psi_{t'}}$ to be the Gram matrix corresponding to this set of vectors. Then,
\begin{align}
    \mathrm{per}(G) \geq \begin{cases}
        1,  &T\leq d \\
        \frac{T}{d},  &T>d
    \end{cases}
\end{align}
\end{corollary}
\begin{proof}
    By definition, the Gram matrix is a $T$-by-$T$ Hermitian, PSD matrix that satisfies $G_{tt}=1$. Thus, Lemma~\ref{lem:perm-cor-LB} yields
    \begin{align}
        \mathrm{per}(G)\ge \frac{1}{T}\norm{G}_{F}^2
    \end{align}
We can then lower bound $\norm{G}_F^2$ as follows. First, observe that
\begin{align}
    \norm{G}_F^2 = \tr{G^2}= \sum_{i,j=1}^T |\braket{\psi_i|\psi_j}|^2 = \sum_{i=j}^T \norm{\psi_i}^2+\sum_{i\ne j}^T |\braket{\psi_i |\psi_j}|^2 = T+\sum_{i\ne j}^T |\braket{\psi_i |\psi_j}|^2
\end{align}
thus
\begin{align}
    \mathrm{per}(G) \ge \frac{1}{T}\norm{G}_F^2 = 1 + \frac{1}{T}\sum_{i\ne j}^T |\braket{\psi_i |\psi_j}|^2.
\end{align}
Since the above right hand term is always non-negative, for $T\le d$, we obtain the unit lower bound by dropping this term. Note that an orthonormal set of vectors will result in this term becoming zero, so the unit lower bound is the best we can obtain via Lemma~\ref{lem:perm-cor-LB} in this regime and is further tight since the permanent of a Gram matrix of orthonormal vectors is 1. Now, suppose that $T>d$, then the term on the right-hand side must be strictly positive. Recall that $\tr{G}=T=\sum_{k=1}^r \lambda_k$ where $r=\mathrm{rank}(G)=\mathrm{span}_\mathbb{C}\cur{\ket{\psi_t}}$ and $\lambda_k$ are the eigenvalues of $G$. Then observe that
\begin{align}
    T^2=\tr{G}^2=\pren{\sum_{k=1}^r 1\cdot \lambda_k}^2 \le \pren{\sum_{k=1}^r 1}\pren{\sum_{k=1}^r \lambda_k^2}=r\norm{G}_F^2\le d\norm{G}_F^2,
\end{align}
where we first use Cauchy-Schwarz and then the fact that $r\leq d$. Combining this with Lemma~\ref{lem:perm-cor-LB}, we obtain $\mathrm{per}(G) \geq \frac{T}{d}$.
\end{proof}
As a note, the above inequalities lower bounding $\norm{G}_F^2$ can be saturated in the following way: first $r=d$ given that $\mathrm{span}_\mathbb{C}\cur{\ket{\psi_t}}=\HC$, i.e., the vectors span the full space. Second, Cauchy-Schwarz is saturated whenever the vectors are a scalar multiple of each other which implies here that $\lambda_k = c1$ for all $k$. We can directly solve for $c$ using the left hand side of the above equation leading to $c=\lambda_k=T/d$. Recall that the Gram matrix can be rewritten as $G=X^\dagger X$ where $X=\brak{\ket{\psi_1}\ \cdots \ket{\psi_t}}$ is the $d\times T$ matrix with $\ket{\psi_t}$ as the $t$-th column. Now define the $d\times d$ matrix $S=XX^\dagger= \sum_{t=1}^T \ketbra{\psi_t}{\psi_t}$ called the frame operator. Then by the singular value theorem for $X=U\Sigma V^\dagger$ we have that $G=V \Sigma^\dagger \Sigma V^\dagger$ and $S=U\Sigma \Sigma^\dagger U^\dagger$, thus $G$ and $S$ have the same spectrum (up to zero entries). Crucially, we have that the above inequality saturates if and only if $S=\tfrac{T}{d}\mathbb{I}$ which occurs if and only if $\{\ket{\psi_t}\}_{t=1}^T$ forms a tight frame (by Thm 1.3.1 of~\cite{christensen_introduction_2016}, where Chapter 1 gives an introduction to frames). Tight frames have many important implications for POVMs in QIT (see~\cite{scott_tight_2006}). 

The remainder of our paper is organized as follows. In Sec.~\ref{sec:distinguishing-LBs} we prove the hardness of several many-to-one distinguishing tasks that we will, in Sec.~\ref{sec:product-testing-LBs}, use to prove lower bounds on the sample complexity of product testing algorithms that only utilize single-copy measurements. Finally, in Sec.~\ref{sec:MP-testing-alg}, we provide a simple algorithm for product testing with single-copy, local measurements (i.e., measurements only on individual qudits) that is optimal for $n=2$ and, we conjecture, optimal up to log factors for all $n>2$.

\section{Hardness of Distinguishing Random Pure States from the Maximally Mixed State} \label{sec:distinguishing-LBs}
In this section, we will prove several results regarding the hardness of distinguishing certain ensembles of random states from the maximally mixed state when restricted to single-copy measurements. In our case, we are interested in these results because they will imply sample lower bounds on the entanglement testing tasks of interest. However, the technical lemmas may be of independent interest. We begin by revisiting a distinguishing problem used in Ref.~\cite{chen2022Exponential} to prove an exponential lower bound for purity testing.

\subsection{Hardness of Distinguishing a Random Pure State from Maximally Mixed}

It was shown in Ref.~\cite{chen2022Exponential} that distinguishing a Haar random pure state from the maximally mixed state requires exponentially many samples. Because our main results will utilize this result, and because it provides intuition for our generalization of the result, we will re-state the lemma and provide an alternative proof by utilizing a lower bound on the permanent of a Gram matrix~\cite{grone1988Permanental}. This technique will make the conditions for equality explicit, which we will utilize to show that the lower bound cannot be tightened further even if one restricts their single-copy measurements further (e.g. by forcing them to be fully-product POVMs). Our proof simplifies the proof of a technical lemma from Ref.~\cite{chen2022Exponential}.
\begin{lemma}[Lemma 5.12 in Ref.~\cite{chen2022Exponential}] \label{lem:perm-overlap-LB}
    For any collection of pure states $\ket{\psi_1}, \dots, \ket{\psi_T} \in \mathbb{C}^d$, 
    \begin{align}
        \sum_{\pi \in S_T} \tr{P_d(\pi) \bigotimes_{t=1}^T \ketbra{\psi_t}{\psi_t} } \ge 1.
    \end{align}
\end{lemma}
\begin{proof}
    Observe
    \begin{align}
        \sum_{\pi \in S_T} \tr{P_d(\pi) \bigotimes_{t=1}^T \ketbra{\psi_t}{\psi_t} } &=   \sum_{\pi \in S_T} \tr{ \bigotimes_{t=1}^T \ketbra{\psi_{\pi^{-1}(t)}}{\psi_t} }, \quad &\text{Eq.}~\eqref{eq:action-of-sym-group}\\ 
        &= \sum_{\pi \in S_T} \prod_{t=1}^T \braket{\psi_t | \psi_{\pi^{-1}(t)}},\quad &\text{cyclicity}\\
        &= \sum_{\pi \in S_T} \prod_{t=1}^T G_{t \pi^{-1}(t)}, \quad &G_{ij}\coloneqq \braket{\psi_i | \psi_j} \\
        &= \mathrm{per}(G), \quad &\text{Def.}~\ref{def:permanent}
    \end{align}
    where we have identified the penultimate line as the permanent of a the Gram matrix of $\{\ket{\psi_t}\}_{t=1}^T$. The Gram matrix is a $T$-by-$T$ positive semi-definite, Hermitian matrix with unit diagonal (i.e., $G_{tt} = \braket{\psi_t|\psi_t}=1$), thus by Lemma~\ref{lem:perm-cor-LB} we have that
     \begin{align}
        \sum_{\pi \in S_T} \tr{P_d(\pi) \bigotimes_{t=1}^T \ketbra{\psi_t}{\psi_t} } &=  \mathrm{per}(G) \ge 1,
    \end{align}
    as desired. Moreover, note that equality can always be reached when $T \le d$ by choosing $\{\ket{\psi_t}\}$ to be an orthonormal set in $\mathbb{C}^d$. In this case $G_{ij} = \delta_{ij}$ and equality can be reached as long as $T \le d$.
\end{proof}
With this technical lemma in place, we can re-state one of the main theorems from Ref.~\cite{chen2022Exponential}. The result simply states that, using single-copy measurements, the sample complexity of any algorithm used to distinguish between a Haar random pure state and the maximally mixed state must scale with the dimension of the state. If one had a single-copy purity tester, they could solve this distinguishing task, thus the lower bound on the distinguishing task immediately implies a lower bound on any single-copy purity tester. In our case, we will be using this as an intermediate step towards proving lower bounds on entanglement testing tasks; however, our careful analysis of this result and the relevant related literature led to an open question that we will highlight below.
\begin{proposition}[Random Pure State versus Maximally Mixed State, Thm. 5.11~\cite{chen2022Exponential}] \label{prop:max-mixed-vs-haar} \index{Purity Testing Lower Bound}
    Any learning algorithm utilizing (potentially adaptive) single-copy measurements requires
    \begin{align}
        T \geq \Omega(d^{1/2}) 
    \end{align}
    copies of $\rho\in \DC(\mathbb{C}^d)$ to distinguish (with probability at least 2/3) whether $\rho$ is a Haar random pure state or the maximally mixed state. Equivalently, one can say any such algorithm will have bias $b \le O\pren{\frac{T^2}{d}}$.
\end{proposition}
\begin{proof}
    We will consider the many-versus-one distinguishing task with null and alternative hypotheses given as
    \begin{align}
       \rho_{mm} = \frac{\mathbb{I}}{d} \quad \text{versus} \quad \{ \ketbra{v}{v} = V \ketbra{0}{0} V^{\dagger}: V \sim \UC_d  \},
    \end{align}
    where we recall from Def.~\ref{def:Haar-random-state} that $\ket{v}$ is simply a Haar random pure state. After $T$ rounds of single-copy measurements, we would like to distinguish between $p^{\rho_{mm}}(\ell)$ and $\mathbb{E}_v[p^{\ketbra{v}{v}}(\ell)]$, with high probability. To show that this is impossible unless we take $T$ to be sufficiently large, we must lower bound the one-sided likelihood ratio. Starting with Eq.~\eqref{eq:classical-outcome-distribution}, we can write
    \begin{align}
        \frac{\mathbb{E}_v \left[p^{\ketbra{v}{v}}(\ell)\right]}{p^{\rho_{mm}}(\ell)} &= \frac{\mathbb{E}_v \left[\prod_{t=1}^T \tr{d\cdot a_{s_t}\cdot \ketbra{v}{v} \ketbra{\psi_{s_t}}{\psi_{s_t}}}\right]}{\prod_{t=1}^T \tr{d\cdot a_{s_t}\cdot \frac{\mathbb{I}}{d} \ketbra{\psi_{s_t}}{\psi_{s_t}}}},\\ 
        &= d^T \mathbb{E}_v \left[  \tr{\ketbra{v}{v}^{\otimes T} \bigotimes_{t=1}^T \ketbra{\psi_{s_t}}{\psi_{s_t}} } \right],\\
        &= d^T  \tr{ \mathbb{E}_v[\ketbra{v}{v}^{\otimes T}] \bigotimes_{t=1}^T \ketbra{\psi_{s_t}}{\psi_{s_t}} } ,\\
        &=\frac{d^T}{d(d+1)\dotsm (d+T-1)} \sum_{\pi \in S_T} \tr{P_d(\pi) \bigotimes_{t=1}^T \ketbra{\psi_{s_t}}{\psi_{s_t}} }, \quad &\text{ Lemma~\ref{lem:sym-sub-proj}}\\
        &\geq \frac{d^T}{d(d+1)\dotsm (d+T-1)},  \quad &\text{Lemma~\ref{lem:perm-overlap-LB}}\\
        &= \prod_{t=0}^{T-1} \frac{d}{d+t},\\
        &= \prod_{t=0}^{T-1} \left(1+\frac{t}{d}\right)^{-1},\\
        &= \exp\left(\log\left(\prod_{t=0}^{T-1} \left(1+\frac{t}{d}\right)^{-1}\right)\right),\\
        &= \exp\left(-\sum_{t=0}^{T-1}\log\left(1+\frac{t}{d}\right)\right),\\
        &\geq \exp\left(-\sum_{t=0}^{T-1}\frac{t}{d}\right),\quad \log(1+x)\le x\\
        &= \exp\left(-\frac{T(T-1)}{2d}\right)\\
        &\geq 1-\frac{T(T-1)}{2d}, \quad \exp(-x)\geq 1-x\\
        \implies  \mathbb{E}_v \left[\frac{p^{\ketbra{v}{v}}(\ell)}{p^{\rho_{mm}}(\ell)}\right] 
        &\geq 1-\frac{T(T-1)}{2d},
    \end{align}
    Thus, by Lemma~\ref{lem:one-sided-LeCam}, we have that 
    \begin{align}
        d_{TV}(\mathbb{E}_v p^{\rho_{\ketbra{v}{v}}}(\ell),p^{\rho_{mm}}(\ell))\leq \frac{T(T-1)}{2d} = O\left(\frac{T^2}{d}\right). 
    \end{align}
    If an algorithm to distinguish these ensembles exists, $\Omega(1) \leq p_{\rm succ}$. Together, with the above upper bound, we have that $\Omega(1) \leq p_{\rm succ} \leq O(T^2/d)$, which can only be true if we take
    \begin{align}
        T \geq \Omega(d^{1/2}),
    \end{align}
    as desired.
    \end{proof}
    \subsubsection{A Few Remarks on Purity Testing with Single-copy Measurements}
    In Ref.~\cite{chen2022Exponential}, the authors also construct a matching single-copy algorithm for distinguishing these two ensembles which is based on a standard \textit{collision tester} from classical statistical learning theory (see Theorems 5.13 and 5.14 in Ref.~\cite{chen2022Exponential} and the references therein). This proves that this distinguishing task cannot be used to prove a larger lower bound on purity testing; however, their algorithm does not actually solve the purity testing problem in the standard sense (i.e., test whether a state is pure or $\varepsilon$-far from any pure state~\cite{montanaro2016Survey}). For this task, their collision-based algorithm breaks down. One cannot achieve the necessary anti-concentration between the two distributions when one state is pure and the other is mixed but (potentially) far from maximally mixed. 
    
    Of course, the recent algorithms for \textit{purity estimation} with single-copy, global measurements~\cite{anshu2022Distributed,gong2024sample} could be used to solve the standard purity testing problem; however, these may not be sample-optimal because estimation is at least as hard as testing. Moreover, while single-copy, global measurements are certainly more feasible on near-term devices than multi-copy measurements, the random measurements needed for the collision testers in Refs,~\cite{anshu2022Distributed,gong2024sample,chen2024Optimala} still remain a serious technical challenge on near-term platforms. To make these tasks feasible on today's multi-qubit processors, it would be preferable to have an algorithm for purity testing that utilizes only \textit{single-copy, local} measurements on individual qubits. As we see from the equality conditions in Lemma~\ref{lem:perm-overlap-LB}, unless the sample complexity of such a task is larger than $2^n$, we cannot tighten the lower bound above. As such, we identify the following open question.

    \begin{openquestion}[Single-copy, local purity testing] Does there exist an algorithm using only single-copy, local measurements that can solve the purity testing problem using $O(2^{n/2})$ samples? If not, can alternative methods yield a lower bound larger than $\Omega(2^{n/2})$ when enforcing the restriction that all measurement operators are fully product?
    \end{openquestion}

    We will say more regarding this and other open problems we identify before concluding the paper. For now, we turn to another distinguishing task that is essential for our lower bound on single-copy BP testing algorithms.

\subsection{Hardness of Distinguishing a Random Product State from Maximally Mixed}
To prove an exponential lower bound for product testing with single-copy measurements, we need to generalize Lemma~\ref{lem:perm-overlap-LB} to collections of pure, multi-qudit states. Here, we will prove the $n=2$ case which contains all of the structure of the proof for all $n>2$.
\begin{lemma}\label{lem:prod-of-perms-LB}
 For any collection of pure states $\ket{\psi_1}^{AB}, \dots, \ket{\psi_T}^{AB} \in \mathbb{C}^{d} \tilde{\otimes} \mathbb{C}^d$, with $\psi^{AB}_t \coloneqq  \ketbra{\psi_t}{\psi_t}^{AB}$, we have
    \begin{align}
         \sum_{ \pi, \sigma \in \SC_T } \tr{P_{d}(\pi) \tilde{\otimes} P_{d}(\sigma) \bigotimes_{t=1}^T \psi^{AB}_t } \geq 1,
    \end{align}
    where the $\tilde{\otimes}$ represents the tensor product between $\HC_A$ and $\HC_B$, and $\otimes$ represents the tensor product between the $T$ isomorphic copies of these spaces.
\end{lemma}
\begin{proof}
    We will prove this claim by induction on $T$. Consider first the base case when $T=2$. We have
    \begin{align}
        \sum_{\pi, \sigma \in \SC_2} \tr{P_{d}(\pi) \tilde{\otimes} P_{d}(\sigma) \psi^{AB}_1 \otimes \psi^{AB}_2 } &= \tr{(\mathbb{I}\tilde{\otimes}\mathbb{I} + \mathbb{F}_{12} \tilde{\otimes} \mathbb{I} + \mathbb{I} \tilde{\otimes} \mathbb{F}_{12} + \mathbb{F}_{12} \tilde{\otimes} \mathbb{F}_{12} )\cdot \psi^{AB}_1 \otimes \psi^{AB}_2},\\
        &= 1 + \tr{\psi_1^A \psi_2^A} +  \tr{\psi_1^B \psi_2^B} +  \tr{\psi_1^{AB} \psi_2^{AB}},\\
        &\geq 1,
    \end{align}
    because each of the three terms in the second line are the trace of a product of PSD operators and thus non-negative. 
    Now, let us assume the claim holds for a collection of $T$ pure states and proceed to show it holds for $T+1$. First, note that any permutation in $\SC_{T}$ can be viewed as an element of $\SC_{T+1}$ with one tensor factor fixed. In particular, we can view $\mathbb{I} \otimes \Pi_{\rm sym}^{d,T}$ (which acts on $\mathbb{C}^{d} \otimes (\mathbb{C}^{d} )^{\otimes T}$) as a sum over all permutations in $\SC_{T+1}$ that fix the first tensor factor. Recalling from Sec.~\ref{sec:rep-theory-essentials} that $P_d(\pi)\Pi_{\rm sym}^{d,T+1} = \Pi_{\rm sym}^{d,T+1} P_d(\pi) = \Pi_{\rm sym}^{d,T+1}$ holds for any permutation $\pi \in \SC_{T+1}$, and using $\Pi_{\rm sym}^{d,T+1} = \frac{1}{(T+1)!}\sum_{\pi\in \SC_{T+1}}P_d(\pi)$, we have that 
    \begin{align}
         \left(\mathbb{I} \otimes \Pi_{\rm sym}^{d,T} \right) \left(\sum_{\pi \in \SC_{T+1}} P_d(\pi) \right) \left( \mathbb{I} \otimes \Pi_{\rm sym}^{d,T} \right) &= \sum_{\pi \in \SC_{T+1}} P_d(\pi), 
    \end{align}
    from which it follows that
    \begin{align} \label{eq: invariance}
        (\mathbb{I} \otimes \Pi_{\rm sym}^{d,T}) \tilde{\otimes } (\mathbb{I}\ot \Pi_{\rm sym}^{d,T}) \left(\sum_{\pi, \sigma \in \SC_{T+1}} P_d(\pi) \tilde{\otimes} P_d(\sigma) \right)   (\mathbb{I} \otimes \Pi_{\rm sym}^{d,T}) \tilde{\otimes } (\mathbb{I}\ot \Pi_{\rm sym}^{d,T}) &= \sum_{\pi, \sigma \in \SC_{T+1}} P_d(\pi) \tilde{\otimes} P_d(\sigma). 
    \end{align}
    In order to leverage this invariance, let us define $\Psi \coloneqq \bigotimes_{t=2}^{T+1} \psi_t^{AB}$. Then, observe
    \begin{align}
        &\sum_{\pi,\sigma \in \SC_{T+1}} \tr{ \left(P_{d}(\pi) \tilde{\otimes} P_d(\sigma)\right) \cdot \left(\psi_1^{AB} \otimes \Psi \right)} =  \tr{  \left(\sum_{\pi,\sigma \in \SC_{T+1}} P_{d}(\pi) \tilde{\otimes} P_d(\sigma) \right) \cdot \psi_1^{AB} \otimes \Psi},\\
        &=  \tr{    (\mathbb{I} \otimes \Pi_{\rm sym}^{d,T}) \tilde{\otimes } (\mathbb{I}\ot \Pi_{\rm sym}^{d,T}) \left(\sum_{\pi, \sigma \in \SC_{T+1}} P_d(\pi) \tilde{\otimes} P_d(\sigma) \right) (\mathbb{I} \otimes \Pi_{\rm sym}^{d,T}) \tilde{\otimes } (\mathbb{I}\ot \Pi_{\rm sym}^{d,T}) \cdot \psi_1^{AB} \otimes \Psi},\\
        &=  \sum_{\pi, \sigma \in \SC_{T+1}} \tr{  P_d(\pi) \tilde{\otimes} P_d(\sigma) \left(  (\mathbb{I} \otimes \Pi_{\rm sym}^{d,T}) \tilde{\otimes } (\mathbb{I}\ot \Pi_{\rm sym}^{d,T}) \cdot (\psi_1^{AB} \otimes \Psi) \cdot (\mathbb{I} \otimes \Pi_{\rm sym}^{d,T}) \tilde{\otimes } (\mathbb{I}\ot \Pi_{\rm sym}^{d,T}) \right)},\\ 
        &=   \sum_{\pi, \sigma \in \SC_{T+1}} \tr{  P_d(\pi) \tilde{\otimes} P_d(\sigma) \cdot \psi_1^{AB} \otimes \overline{\Psi}}, 
    \end{align}
    where we utilized the cyclicity of trace in the third equality and then defined the symmetrized state $ \overline{\Psi} \coloneqq  \Pi_{\rm sym}^{d,T} \tilde{\otimes} \Pi_{\rm sym}^{d,T} \cdot \Psi \cdot \Pi_{\rm sym}^{d,T} \tilde{\otimes } \Pi_{\rm sym}^{d,T}$. Crucially,  this will allow us to absorb permutations into the symmetrized state due to the fact that for all $\alpha,\beta \in \SC_T$, we have $P_d(\alpha) \tilde{\otimes}  P_{d}(\beta) \cdot\overline{\Psi} = \overline{\Psi}.$ Now, we split the sum into all terms where both $\pi$ and $\sigma$ fix the first tensor factor and all of the rest.
     \begin{align}
         \sum_{ \pi, \sigma \in \SC_{T+1} } \tr{P_{d}(\pi) \tilde{\otimes} P_{d}(\sigma) \cdot \psi_1^{AB} \otimes \overline{\Psi}} &=  \sum_{ \substack{\pi, \sigma \in \SC_{T+1} \\ \pi(1)=\sigma(1)=1}} \tr{P_{d}(\pi) \tilde{\otimes} P_{d}(\sigma) \cdot \psi_1^{AB} \otimes \overline{\Psi}}\\
         &+  \sum_{ \substack{\pi, \sigma \in \SC_T \\ \pi(1)\neq 1~\vee~\sigma(1)\neq 1}} \tr{P_{d}(\pi) \tilde{\otimes} P_{d}(\sigma) \cdot \psi_1^{AB} \otimes \overline{\Psi}}.
    \end{align}
    We will show the first set of terms sum at least to unity while the remaining terms are non-negative. 

    \noindent \textbf{Case 1} ($\pi(1)=\sigma(1)=1$): Here is where we will utilize the inductive hypothesis.
    For a permutation $\pi\in \SC_{T+1}$ with $\pi(1)=1$ we write $\pi = (1)\pi'$ (in disjoint cycle notation) with $\pi'\in\SC_T$ permuting the symbols $\lbrace 2,\dots,T+1\rbrace$, and similarly for $\sigma\in \SC_{T+1}$ with $\sigma(1)=1$.
    We then have
    \begin{align}
         \sum_{ \substack{\pi, \sigma \in \SC_{T+1} \\ \pi(1)=\sigma(1)=1}} \tr{P_{d}(\pi) \tilde{\otimes} P_{d}(\sigma) \cdot \psi_1^{AB} \otimes \overline{\Psi}} 
         &= \sum_{\pi', \sigma' \in \SC_{T}} \tr{\left(\mathbb{I} \otimes P_d(\pi') \right)\tilde{\otimes} \left(\mathbb{I} \otimes P_{d}( \sigma')\right) \cdot \psi_1^{AB} \otimes \overline{\Psi}},\\
         &= \sum_{\pi', \sigma' \in \SC_{T}} \tr{ \psi_1^{AB}} \cdot \tr{ P_d(\pi') \tilde{\otimes}  P_{d}(\sigma') \cdot\overline{\Psi}},\\
         &= \sum_{\pi', \sigma' \in \SC_{T}} \tr{\overline{\Psi}},\\
         &= (T!)^2\cdot \tr{\overline{\Psi}},\\
         &= (T!)^2\cdot \tr{\Pi_{\rm sym}^{d,T} \tilde{\otimes} \Pi_{\rm sym}^{d,T} \cdot \Psi \cdot \Pi_{\rm sym}^{d,T} \tilde{\otimes } \Pi_{\rm sym}^{d,T}},\\
         &= (T!)^2\cdot \tr{\Pi_{\rm sym}^{d,T} \tilde{\otimes} \Pi_{\rm sym}^{d,T} \cdot \Psi},\\
         &= (T!)^2\cdot \frac{1}{(T!)^2}  \sum_{\pi', \sigma' \in \SC_{T}} \tr{P_d(\pi') \tilde{\otimes} P_d(\sigma') \cdot \Psi},\\
         &=\sum_{\pi', \sigma' \in \SC_{T}} \tr{P_d(\pi') \tilde{\otimes} P_d(\sigma') \cdot \bigotimes_{t=2}^{T+1} \psi^{AB}_t},\\
         &\geq 1,
    \end{align}
    by the inductive hypothesis.  It remains to show that the other terms in the above sum are always non-negative. \\

    \noindent \textbf{Case 2} ($\pi(1) \neq 1 \vee \sigma(1) \neq 1$): To show that this remaining sum is non-negative note that Lemma~\ref{lemma:double-coset-decomp} guarantees that any $\pi, \sigma \in \SC_{T+1}$ can be expressed (in disjoint cycle notation) as 
    \begin{align}
        \pi &= \alpha (1~2)^a \beta \quad \text{and} \quad \sigma = \gamma (1~2)^b \delta,
    \end{align}
    where $\alpha, \beta, \gamma, \delta \in \SC_T^{[2,T+1]} \subset \SC_{T+1}$ (i.e., permutations in $\SC_{T+1}$ acting non-trivially only on systems $t=2$ to $t=T+1$) and $a,b \in \{0,1\}$. Moreover, by assumption, $a$ and $b$ are not simultaneously zero. With this notation in mind, the sum becomes
    \begin{align}
        &\sum_{ \substack{\pi, \sigma \in \SC_T \\ \pi(1)\neq 1~\vee~\sigma(1)\neq 1}} \tr{P_{d}(\pi) \tilde{\otimes} P_{d}(\sigma) \cdot \psi_1^{AB} \otimes \overline{\Psi}} \\
        &= \sum_{ \substack{\pi, \sigma \in \SC_T \\ \pi(1)\neq 1~\vee~\sigma(1)\neq 1}} \tr{P_{d}(\alpha (1~2)^a \beta) \tilde{\otimes} P_{d}(\gamma (1~2)^b \delta) \cdot \psi_1^{AB} \otimes \overline{\Psi}},\\
        &= \sum_{ \substack{\pi, \sigma \in \SC_T \\ \pi(1)\neq 1~\vee~\sigma(1)\neq 1}} \tr{\left(P_{d}(\alpha) \tilde{\otimes} P_{d}(\gamma) \right) \cdot \left( P_d((1~2)^a) \tilde{\otimes} P_d((1~2)^b) \right)\cdot \left( P_{d}(\beta) \tilde{\otimes} P_{d}(\delta) \right) \cdot \psi_1^{AB} \otimes \overline{\Psi}},\\
        &= \sum_{ \substack{\pi, \sigma \in \SC_T \\ \pi(1)\neq 1~\vee~\sigma(1)\neq 1}} \tr{P_{d}((1~2)^a) \tilde{\otimes} P_{d}((1~2)^b) \cdot \psi_1^{AB} \otimes \left( P_d(\beta) \tilde{\otimes} P_d(\delta) \cdot \overline{\Psi} \cdot  P_d(\alpha) \tilde{\otimes} P_d(\gamma) \right)}, \\
        &= \sum_{ \substack{\pi, \sigma \in \SC_T \\ \pi(1)\neq 1~\vee~\sigma(1)\neq 1}} \tr{P_{d}((1~2)^a) \tilde{\otimes} P_{d}((1~2)^b) \cdot \psi_1^{AB} \otimes  \overline{\Psi} },\\
        &\geq 0,
    \end{align}
    where the second equality follows from the homomorphism property of the representation (i.e., $P_d(\pi \sigma) = P_d(\pi)P_d(\sigma)$), the third from cyclicity of trace, and the last line from the fact that
    \begin{align}
        \tr{P_{d}((1~2)^a) \tilde{\otimes} P_{d}((1~2)^b) \cdot \psi_1^{AB} \otimes  \overline{\Psi} } &= \begin{cases}
            \tr{\psi^{A}_1 \overline{\Psi}^A_2}, \quad  & a=1, b=0\\
            \tr{\psi^{B}_1 \overline{\Psi}^B_2}, \quad  & a=0, b=1\\
            \tr{\psi_1 \overline{\Psi}_2}, \quad  & a=1, b=1
        \end{cases}
    \end{align}
    where $\overline{\Psi}_2$ is the reduced state after tracing out systems $3,\ldots,T+1$. The equalities follows then follow from the definition of the partial trace and the swap trick (see Lemma~\ref{lemma:swap-trick}). In all three cases, we have the overlap of two PSD operators which is non-negative, thus the whole sum is non-negative. This completes the proof.
\end{proof}
 We note that this lower bound cannot be uniformly improved when $T\le d$. To see that we can saturate it, recall that from the definition of the symmetric subspace projector that Lemma~\ref{lem:prod-of-perms-LB} can be expressed as
\begin{align} \label{eq:prod-sym-proj-functional}
          \tr{ \Pi_{\rm sym}^{d,T} \tilde{\otimes} \Pi_{\rm sym}^{d,T} \bigotimes_{t=1}^T\ketbra{\psi_t}{\psi_t}^{AB}} \geq \frac{1}{(T!)^2}.
    \end{align}
Next, consider a collection of states that are product across the $A:B$. That is, let $\ket{\psi_t}^{AB}=\ket{\psi_t}^A\tilde{\otimes} \ket{\psi'_t}^B$ for all $t \in [T]$. Then, we have
\begin{align}
     &\tr{ \Pi_{\rm sym}^{d,T} \tilde{\otimes} \Pi_{\rm sym}^{d,T} \bigotimes_{t=1}^T\ketbra{\psi_t}{\psi_t}^{AB}} =  \tr{ \Pi_{\rm sym}^{d,T} \tilde{\otimes} \Pi_{\rm sym}^{d,T} \bigotimes_{t=1}^T \left(\ketbra{\psi_t}{\psi_t}^A\tilde{\otimes} \ketbra{\psi_t}{\psi_t}^B\right)},\\
     &= \tr{ \Pi_{\rm sym}^{d,T} \bigotimes_{t=1}^T \ketbra{\psi_t}{\psi_t}^{A}} \cdot \tr{ \Pi_{\rm sym}^{d,T} \bigotimes_{t=1}^T \ketbra{\psi_t}{\psi_t}^{B}},\\
     &=\frac{1}{(T!)^2} \cdot \tr{ \sum_{\pi \in S_T} P_d(\pi) \bigotimes_{t=1}^T \ketbra{\psi_t}{\psi_t}^A} \cdot \tr{ \sum_{\sigma \in S_T} P_d(\sigma) \bigotimes_{t=1}^T \ketbra{\psi_t}{\psi_t}^B},\\
     &\geq \frac{1}{(T!)^2}, 
\end{align}
    where we have used the inequality in Lemma~\ref{lem:perm-cor-LB} which is tight when $\{\ket{\psi_A^t}\}$ and $\{\ket{\psi_B^t}\}$ are orthonormal sets on their respective spaces. Because there exists at least one collection of states for which this equality is reached, we cannot improve the lower bound in general. 
    
This technical lemma allows us to prove a hardness result for distinguishing local Haar random states (i.e., random product states) from global Haar random states.
\begin{proposition}[Random Bipartite Product State versus Maximally Mixed]\label{prop:max-mixed-vs-bipartite-prod}
    Any learning algorithm utilizing (potentially adaptive) single-copy measurements requires 
    \begin{align}
        T\geq \Omega(d^{1/2})
    \end{align}
    copies of $\rho$ to distinguish (with probability at least 2/3) whether $\rho$ is a Haar random bipartite product state or the maximally mixed state on $\mathbb{C}^{d} \tilde{\otimes} \mathbb{C}^{d}$. Equivalently, any such algorithm will have bias $O(\frac{T^2}{d})$.
\end{proposition}
Note that this is actually a looser lower bound than Proposition~\ref{prop:max-mixed-vs-haar} because the total dimension of $\mathbb{C}^d \tilde{\otimes} \mathbb{C}^d$ is $d^2$, so any algorithm distinguishing a global Haar random state from maximally mixed state on that space will achieve bias $O(\frac{T^2}{d^2})$, or equivalently, a sample complexity at least $\Omega(d)$. We now proceed to the proof of the proposition.
\begin{proof}
    We will consider the many-versus-one distinguishing task with null and alternative hypotheses given as
    \begin{align}
        \rho_{mm} = \frac{\mathbb{I}}{d^2} \quad \text{versus} \quad \{\ketbra{v}{v} \tilde{\otimes} \ketbra{w}{w} &= V \ketbra{0}{0} V^{\dagger} \tilde{\otimes} W \ketbra{0}{0} W^{\dagger} : V,W \sim \UC_d\},
    \end{align}
   where $\ket{v}, \ket{w} \in \mathbb{C}^d$ are sampled independently from the Haar measure. Following the same initial steps as in the proof of Proposition~\ref{prop:max-mixed-vs-haar} and suppressing the $AB$ superscript to ease notation, we obtain 
    \begin{align}
       \frac{ \mathbb{E}_{v,w} \left[p^{\ketbra{u}{u} \tilde{\otimes} \ketbra{w}{w}}(\ell)\right]}{p^{\rho_{mm}}(\ell)} &= \frac{\mathbb{E}_{v,w} \left[\prod_{t=1}^T \tr{d^2\cdot a_{s_t}\cdot \ketbra{v}{v} \tilde{\otimes} \ketbra{w}{w} \ketbra{\psi_{s_t}}{\psi_{s_t}}}\right]}{\prod_{t=1}^T \tr{d^2\cdot a_{s_t}\cdot \frac{\mathbb{I}}{d^2} \ketbra{\psi_{s_t}}{\psi_{s_t}}}}, \\
        &= d^{2T} \mathbb{E}_{u,w} \left[\tr{ \ketbra{v}{v} \tilde{\otimes} \ketbra{w}{w} \cdot \bigotimes_{t=1}^T \ketbra{\psi_t}{\psi_t} }\right],\\
        &= d^{2T} \tr{ \mathbb{E}_{v} \left[\ketbra{v}{v}\right] \tilde{\otimes} \mathbb{E}_{w} \left[\ketbra{w}{w}\right] \cdot \bigotimes_{t=1}^T \ketbra{\psi_t}{\psi_t} },\\
        &= d^{2T}  \tr{ \frac{\Pi_{\rm sym}^{d,T}}{\dim{\Pi_{\rm sym}^{d,T}}} \tilde{\otimes} \frac{\Pi_{\rm sym}^{d,T}}{\dim{\Pi_{\rm sym}^{d,T}}} \cdot \bigotimes_{t=1}^T \ketbra{\psi_t}{\psi_t} }, \quad \text{Lemma}~\ref{lem:sym-sub-proj}\\
        &=\frac{d^{2T}}{\left(d(d+1)\dotsm (d+T-1)\right)^2} \cdot \sum_{\pi, \sigma \in \SC_T} \tr{P_{d}(\pi) \tilde{\otimes} P_{d}(\sigma) \bigotimes_{t=1}^T \ketbra{\psi_t}{\psi_t} } ,\\
         &\geq \frac{d^{2T}}{\left(d(d+1)\dotsm (d+T-1)\right)^2}, \\
        \end{align}
        where in the penultimate equality we expanded the definition of the symmetric subspace projectors and in the last line applied Lemma~\ref{lem:prod-of-perms-LB}. The remaining analysis mirrors the proof of Prop.~\ref{prop:max-mixed-vs-haar} and yields
        \begin{align}
        \frac{ \mathbb{E}_{v,w} \left[p^{\ketbra{v}{v} \tilde{\otimes} \ketbra{w}{w}}(\ell)\right]}{p^{\rho_{mm}}(\ell)} &\geq 1-\frac{T(T-1)}{d}.
    \end{align}
    Again, by Lemma~\ref{lem:one-sided-LeCam}, this implies 
    \begin{align}
        d_{\rm TV}\left(\mathbb{E}_{v,w} \left[p^{\ketbra{v}{v} \tilde{\otimes} \ketbra{w}{w}}(\ell) \right], p^{\rho_{mm}}(\ell)\right) \leq O \left( \frac{T^2}{d}\right).
    \end{align}
    Together with the assumption that an algorithm to distinguish these ensembles with constant bias exists, yields the inequalities $\Omega(1) \leq p_{\rm succ} \leq O \left( \frac{T^2}{d}\right)$, which can only be true if
    \begin{align}
        T \geq \Omega(d^{1/2}).
    \end{align}
\end{proof}
%$T\le1+\left\lfloor\left(\sqrt{1+\sqrt{2}}-1\right)d\right\rfloor$
This result will allow us to prove a lower bound on BP testing. To prove a lower bound on MP testing, we need to extend this analysis to $n>2$, starting with a generalization of Lemma~\ref{lem:prod-of-perms-LB}. We state the multipartite case without proof because it is essentially equivalent to the proof of the $n=2$ case.
\begin{lemma}\label{lem:multi-qudit-overlap-LB}
 For any collection of pure $n$-qudit states $\ket{\psi_1}^{A_1\dotsm A_n}, \dots, \ket{\psi_T}^{A_1 \dotsm A_n} \in  (\mathbb{C}^d)^{\tilde{\otimes} n}$, with $\psi_t^{A_1 \dotsm A_n} \coloneqq  \ketbra{\psi_t}{\psi_t}^{A_1 \dotsm A_n}$, we have
    \begin{align}
         \sum_{ \pi_1, \dots, \pi_n \in \SC_T } \tr{ \widetilde{\bigotimes}_{i=1}^n P_{d}(\pi_i) \bigotimes_{t=1}^T \psi^{A_1 \dotsm A_n}_t } \geq 1,
    \end{align}
    where the $\tilde{\otimes}$ represents the tensor product between the $n$ copies of $\mathbb{C}^d$ and the $\otimes$ the tensor product between the $T$ copies of $(\mathbb{C}^d)^{\tilde{\otimes} n}$.
\end{lemma}
This, in turn, allows us to generalize Proposition~\ref{prop:max-mixed-vs-bipartite-prod} to the multipartite setting. The proof is essentially equivalent to that of Proposition~\ref{prop:max-mixed-vs-bipartite-prod}, so we omit it.
\begin{proposition}[Random multipartite product state versus maximally mixed]\label{prop:max-mixed-vs-multipartite-prod}
    Any learning algorithm utilizing (potentially adaptive) single-copy measurements requires 
    \begin{align}
        T\geq \Omega\left(\sqrt{d/n}\right)
    \end{align}
    copies of $\rho$ to distinguish (with probability at least 2/3) whether $\rho$ is a multipartite product state comprised of independently sampled, local Haar random states or the maximally mixed state on $(\mathbb{C}^{d})^{\tilde{\otimes}n}$. Equivalently, any such algorithm will have bias $O(\frac{nT^2}{d})$.
\end{proposition}
A shortcoming of this lower bound is that it decays with $n$ and is thus only non-trivial when $n \ll d$. To prove a larger lower bound using these methods, one would need a state that is harder to distinguish from a random product state. However, constructing an ensemble of such states for which the analysis can still be carried out seems somewhat non-trivial, thus we leave this as another open question in Sec.~\ref{sec:product-testing-LBs}.

\section{Hardness of Product Testing with Single-copy Measurements} \label{sec:product-testing-LBs}
With these distinguishing results in place, we are prepared to prove our lower bounds on the sample complexity of product testers that utilize only single-copy measurements. To do so, we must formally prove that the ensembles above could be distinguished, with high probability, by the appropriate product tester, should one exist. 

Let $\ket{\psi} \in \mathbb{C}^{d} \tilde{\otimes} \mathbb{C}^d$ and suppose we have a single-copy product tester $\{M,\mathbb{I}-M\}$. From Def.~\ref{def:eps-tester}, we recall that this is taken to imply that $\tr{M \psi^{\otimes T}} \geq 2/3$ for all $\psi \in \MC_2$ and $\tr{M \psi^{\otimes T}} \leq 1/3$ for all $\psi$ that are $\varepsilon$-far from $\MC_2$. We wish to prove that such a tester could distinguish between the following two ensembles with at least constant probability
\begin{align}
    \EC_1 = \{U \ketbra{0}{0} U^{\dagger} : U \sim \UC_{d^2}\} \quad \text{and} \quad \EC_2 = \{V \ketbra{0}{0} V^{\dagger} \tilde{\otimes} W \ketbra{0}{0} W^{\dagger}  : V,W \sim \UC_{d}\}.
\end{align}
Clearly $\EC_2$ is product by construction, so all we need to show is that nearly all states in $\EC_1$ are at least $\varepsilon$-far from product. To see this, observe
    \begin{align}
        \underset{\ket{\psi} \in \mathbb{C}^{d^2}}{\mathbb{E}} \left[\tr{M \psi^{\otimes T}}\right] 
        &= \underset{\ket{\psi} \in \mathbb{C}^{d^2}}{\mathbb{E}} \left[\tr{M \psi^{\otimes T}} \big\rvert \psi~\mathrm{far}\right] \cdot \mathrm{Pr}\left[ \psi~\mathrm{far} \right] +  \underset{\ket{\psi} \in \mathbb{C}^{d^2}}{\mathbb{E}} \left[\tr{M \psi^{\otimes T}} \big\rvert \psi~\mathrm{close}\right] \cdot \mathrm{Pr}\left[ \psi~\mathrm{close} \right],\\
        &\le \frac{1}{3} \cdot \mathrm{Pr}\left[ \psi~\mathrm{far} \right] +  1 \cdot \mathrm{Pr}\left[ \psi~\mathrm{close} \right],\\
        &= \frac{1}{3} \cdot (1-\mathrm{Pr}\left[ \psi~\mathrm{close} \right]) +  1 \cdot \mathrm{Pr}\left[ \psi~\mathrm{close} \right],\\
        &= \frac{1}{3} + \frac{2}{3} \cdot \mathrm{Pr}\left[ \psi~\mathrm{close} \right],\\
        &\le \frac{1}{3} + o(1),
    \end{align}
where the last line follows from Proposition~\ref{prop:far-from-BP}, which shows that the probability of obtaining a state that is close to a product state when sampling uniformly over $\mathbb{C}^{d^2}$ decays doubly exponentially. Thus, if a product tester existed, it could distinguish between these two ensembles with high probability.

For completeness, we prove this essential proposition; however, we note that such results have been proven in many different forms throughout the literature~ \cite{hayden2006Aspects, harrow2013Church, aubrun2017Alice}. For our purposes, it will suffice to utilize the following bound on the maximal Schmidt coefficient of a Haar random pure state from Ref.~\cite{aubrun2017Alice} couples with a union bound over all possible bipartitions. 

\begin{lemma}[Proposition 6.36, Ref.~\cite{aubrun2017Alice}] \label{lem:LB-max-Schmidt-coeff} Let $d_1 \leq d_2$ and consider a Haar random pure state on $\mathbb{C}^{d_1} \otimes \mathbb{C}^{d_2}$ with Schmidt coefficients $\lambda_1 (\psi) \geq \dotsm \geq \lambda_{d_1} (\psi)$. Then, for all $\delta > 0$, we have 
\begin{align}
    \Pr_{\psi \in \mathbb{C}^{d_1} \otimes \mathbb{C}^{d_2}}\left[\lambda_1 (\psi) \geq \frac{1}{\sqrt{d_1}}+\frac{1+\delta}{\sqrt{d_2}}\right] \leq \exp\left(-d_1 \delta^2\right).
\end{align}
\end{lemma}
    Next, recall the fact that the maximal Schmidt coefficient is equivalent to the maximal overlap among product states (a fact that follows from Def.~\ref{def:schmidt-decomp}, triangle inequality, and Cauchy-Schwarz). That is, for all $\ket{\psi} \in \mathbb{C}^{d_1} \otimes \mathbb{C}^{d_2}$, the maximal Schmidt coefficient is given as 
    \begin{align}
        \lambda_1(\psi) = \max_{\ket{\phi}} \abs{\braket{\phi|\psi}}^2,
    \end{align}
    where $\ket{\phi}=\ket{\phi_1}\otimes \ket{\phi_2}$ for arbitrary states $\ket{\phi_1} \in \mathbb{C}^{d_1}, \ket{\phi_2} \in \mathbb{C}^{d_2}$. Recalling that, for pure states, the trace distance takes the form $d_{\rm tr}(\psi, \phi) = (1- \abs{\braket{\psi|\phi}}^2)^{1/2}$, we may also write
    \begin{align}
        \lambda_1(\psi) = 1- \left(\min_{\phi~\mathrm{product}} d_{\rm tr}(\psi, \phi) \right)^2,
    \end{align}
    thus a lower bound on the maximal Schmidt coefficient will give us an upper bound on the minimal trace distance to a product state. We can use these facts and Lemma~\ref{lem:LB-max-Schmidt-coeff} to prove the following proposition. 

    \begin{proposition}[Most Random States are Far From $\BC_n$]\label{prop:far-from-BP}
        Let $\ket{\psi}$ be a Haar random state on $(\mathbb{C}^d)^{\otimes n}$, then, for fixed $0 < \varepsilon<1-o(1)$, we have
        \begin{align}
            \Pr\left[ \exists~\phi \in \BC_n : d_{\rm tr}(\phi,\psi) \leq \varepsilon \right] \leq \exp\left(- \Omega(d^n)\right).
        \end{align}
    \end{proposition}
    \begin{proof}
        Let $\lambda_1^{(S)}(\psi)$ denote the maximal Schmidt coefficient corresponding to the bipartition $S:S^c$ of $\ket{\psi} \in (\mathbb{C}^{d})^{\otimes n}$. Let $s:=|S|$ so that $|S^c| = n-s$. Next, let $A$ denote the event that a Haar random state $\ket{\psi}$ is $\varepsilon$-close (in trace distance) to $\BC_n$,
\begin{align}
    A := \{\exists~\phi \in \BC_n : d_{\mathrm{tr}}(\psi,\phi) \le \varepsilon \},
\end{align}
and for each bipartition $S:S^c$, define the event
\begin{align}
    E_S := \{\exists~\phi~\text{product across}~S : d_{\mathrm{tr}}(\psi,\phi) \le \varepsilon \}.
\end{align}
With these sets defined, we then have 
\begin{align}
   \Pr\left[ \exists~\phi \in \BC_n : d_{\mathrm{tr}}(\phi,\psi) \leq \varepsilon \right]= \Pr(A) = \Pr\!\left(\bigcup_S E_S\right) \leq \sum_S \Pr[E_S],
\end{align}
where the second inequality is the standard union bound. Thus, the probability that a random state is close to a fully product state is upper bounded by the sum of the probabilities that it is close to a state that is product across a specific cut. It then suffices to upper bound 
        \begin{align}
            \Pr[E_S] = \Pr[\lambda_1^{(S)}(\psi) \geq 1-\varepsilon^2]. 
        \end{align}    
        For all $S$ such that $s \leq \lfloor n/2 \rfloor$, we then have that $d^s \leq d^{n-s}$, so we may apply Lemma~\ref{lem:LB-max-Schmidt-coeff} to this specific cut and lower bound of $1-\eps^2$ on $\lambda_1^{(S)}(\psi)$ to obtain 
         \begin{align}
            \Pr[E_S] &\leq \exp\left(-[d^{s/2} + d^{(n-s)/2}-d^{n/2}(1-\varepsilon^2)]^2\right),\label{eq:es_inequality}
        \end{align}  
        which holds as long as $\varepsilon^2 < 1 - \left( d^{-s/2} + d^{-(n-s)/2} \right) $ (which corresponds to the $\delta >0$ condition in Lemma~\ref{lem:LB-max-Schmidt-coeff}). We can further upper-bound $\Pr[E_S]$ by noting that the RHS of Eq.~\eqref{eq:es_inequality} is monotonically decreasing in $s\in[1,\lfloor n/2 \rfloor]$:
        \begin{align*}
            &\frac{d}{ds} \left[\exp\left(-\left(d^{s/2} + d^{(n-s)/2}-d^{n/2}(1-\varepsilon^2)\right)^2\right)\right]\\&=d^{(n-s)/2} \log (d)
   \left(\left(1-\varepsilon ^2\right)-d^{-s/2}-d^{-(n-s)/2}\right) \exp\left(-\left(d^{s/2} + d^{(n-s)/2}-d^{n/2}(1-\varepsilon^2)\right)^2\right)\left(d^s-d^{n/2}\right)\leq0,
        \end{align*}
        which follows from the assumption that was made for Eq.~\eqref{eq:es_inequality} to hold and since $d^s\leq d^{n/2}$. Since the derivative of the function is negative, then the RHS of the inequality is maximized when $s$ takes the minimal value of $s=1$. Therefore, we can bound 
        \begin{align*}
            \Pr[E_S]\leq \exp\left(-[d^{n/2}(1-\varepsilon^2)-d^{1/2} - d^{(n-1)/2}]^2\right)=\exp\left(-d^n[(1-\varepsilon^2)-d^{-(n-1)/2}-d^{-1/2}]^2\right).
        \end{align*}
        Then, the probability that at least one such event occurs can be upper bounded as
        \begin{align}
             \sum_S \Pr[E_S] &= \sum_{s=1}^{\lfloor n/2 \rfloor} \binom{n}{s} \exp\left(-\left[d^{s/2} + d^{(n-s)/2}-d^{n/2}(1-\varepsilon^2)\right]^2\right),\\
             &\leq \sum_{s=1}^{\lfloor n/2 \rfloor} \binom{n}{s} \exp\left(-d^n\left[(1-\varepsilon^2)-d^{-(n-1)/2}-d^{-1/2}\right]^2\right),\\
             &\leq 2^n \exp\left(-d^n\left[(1-\varepsilon^2)-d^{-(n-1)/2}-d^{-1/2}\right]^2\right),\\
             &=\exp\left(n\ln2-d^n\left[(1-\varepsilon^2)-d^{-(n-1)/2}-d^{-1/2}\right]^2\right)\\
             &=\exp\left(-\Omega(d^n)\right).
        \end{align}
        This holds as long as
        \begin{align*}
            1-\varepsilon^2-d^{-1/2}-d^{-(n-1)/2}>0,\\
            \varepsilon<\sqrt{1-d^{-(n-1)/2}-d^{-1/2}}.
        \end{align*}
        Therefore, we have shown that 
        \begin{align}
            \Pr\left[ \exists~\phi \in \BC_n : d_{\rm tr}(\phi,\psi) \leq \varepsilon \right] \leq \exp\left(- \Omega(d^n)\right),
        \end{align}
    as desired.
    \end{proof}

\begin{corollary}[Most Random States are Far From $\MC_n$] \label{cor:far-from-MP}
     Let $\ket{\psi}$ be a Haar random state on $(\mathbb{C}^d)^{\otimes n}$, then, for fixed $0 < \varepsilon<1-o(1)$, we have
        \begin{align}
            \Pr\left[ \exists~\phi \in \MC_n : d_{\rm tr}(\phi,\psi) \leq \varepsilon \right] \leq \exp\left(- \Omega(d^n)\right).
        \end{align}
\end{corollary}
\begin{proof}
    With $E_S$ defined as in the proof of Proposition~\ref{prop:far-from-BP}, we note that 
    \begin{align}
         \Pr\left[ \exists~\phi \in \MC_n : d_{\rm tr}(\phi,\psi) \leq \varepsilon \right] \leq \Pr\!\left(\bigcap_S E_S\right) \leq \Pr\!\left(\bigcup_S E_S\right) \leq \exp\left(-\Omega(d^n)\right),
    \end{align}
    where the final inequality holds by Proposition~\ref{prop:far-from-BP}.
\end{proof}

    Having shown that such ensembles can be distinguished by the relevant product tester, should one exist, we may now prove our main sample complexity lower bounds on said testers.

\subsection{Exponential Sample Lower Bound for BP Testing with Single-copy Measurements}

First, let us consider the BP testing problem as formalized in Ref.~\cite{harrow2017Sequential,jones2024Testing}. The authors of Ref.~\cite{jones2024Testing} proved a information theoretic lower bound on the sample complexity of this task by considering the following many-to-many distinguishing task. They prove that distinguishing between an ensemble of global Haar random states on $(\mathbb{C}^d)^{\tilde{\otimes}n }$ and an ensemble constructed by first choosing a random $S \subset [n]$ and then constructing a random product state by taking the tensor product of a Haar random state $(\mathbb{C}^d)^{\tilde{\otimes} n-|S| }$ with a Haar random state on $(\mathbb{C}^d)^{\tilde{\otimes} |S| }$ is impossible unless $\Omega(n/\log{n})$ samples are used. This is tight up to log factors in light of the algorithm using $O(n)$ samples in Ref.~\cite{harrow2017Sequential}. To get such a strong lower bound, this added step of randomizing the bipartition or ``hiding the cut'' is necessary because, in the multi-copy setting, if one is told where the cut is, one could simply apply a SWAP test across the specified cut and solve the distinguishing task using $O(1)$ samples. 

However, in the single-copy regime, there is no efficient method of testing purity~\cite{chen2022Exponential}, thus it will suffice to consider the somewhat easier task of distinguishing between a global Haar random state and random product state where each tensor factor was sampled from the Haar measure\footnote{Note that if $n$ is odd, we simply take the largest possible bipartition, i.e., $\lfloor n/2 \rfloor$ qudits in one half and $\lceil n/2 \rceil$ qudits in the other.} on $(\mathbb{C}^d)^{\tilde{\otimes} n/2}.$ By taking the largest possible bipartition, we preclude efficiently distinguishing these ensembles using single-copy purity estimation techniques~\cite{anshu2022Distributed,gong2024sample}. 

    \begin{theorem}[Single-copy Lower Bound on BP Testing]\label{thm:BP-test-LB} Any algorithm using, potentially adaptive, single-copy measurements to test whether a state is product across some cut, or is $\varepsilon$-far from any such state (with probability at least $2/3$) must use at least $\Omega(d^{n/4})$ samples.
    \end{theorem}
    \begin{proof}
    Recall from above that a BP tester, if it exists, could distinguish between
    \begin{align}
        \EC_1 = \{U \ketbra{0}{0} U^{\dagger} : U \sim \UC_{d^n}\} \quad \text{and} \quad \EC_2 = \{V \ketbra{0}{0} V^{\dagger} \tilde{\otimes} W \ketbra{0}{0} W^{\dagger}  : V,W \sim \UC_{d^{n/2}}\},
    \end{align}
    with high probability (i.e., $\Omega(1) \leq p_{\rm succ} $). Then, using Lemma~\ref{lem:LeCam-2pt} to upper bound the success probability, we obtain
    \begin{align}
        \Omega(1) &\leq d_{\rm TV}\left(\mathbb{E}_{u} \left[p^{\ketbra{u}{u}}(\ell) \right],\mathbb{E}_{v,w} \left[p^{\ketbra{v}{v} \tilde{\otimes} \ketbra{w}{w}}(\ell) \right]\right),\\
        &\leq d_{\rm TV}\left(\mathbb{E}_{u} \left[p^{\ketbra{u}{u}}(\ell) \right],p^{\rho_{mm}}(\ell)\right) + d_{\rm TV}\left(p^{\rho_{mm}}(\ell),\mathbb{E}_{v,w} \left[p^{\ketbra{v}{v} \tilde{\otimes} \ketbra{w}{w}}(\ell) \right]\right),\\
        &\leq O\left(\frac{T^2}{d^{n}}\right) + O\left(\frac{T^2}{d^{n/2}}\right),
    \end{align}
    which follows from the triangle inequality and then both Prop.~\ref{prop:max-mixed-vs-haar} and Prop.~\ref{prop:max-mixed-vs-bipartite-prod}. Asymptotically, the first term become negligible, so
    \begin{align}
        \Omega(1) \leq O\left(\frac{T^2}{d^{n/2}}\right) \implies T \geq \Omega \left( d^{n/4}\right),
    \end{align}
    as desired.
    \end{proof}
    With unrestricted access to $T$ copies of $\ket{\psi} \in (\mathbb{C}^d)^{\tilde{\otimes} n}$, showed that $T = O(n/\varepsilon^2)$ samples suffice to solve the BP testing problem. Thus, our lower bound establishes an exponential separation between BP testing algorithms with and without the ability to perform multi-copy measurements. Moreover, as mentioned in the introduction, our lower bound immediately implies sample lower bounds on locating un-entanglement~\cite{bouland2024state} and computing the generalized geometric measure of entanglement (see \cite{jones2024Testing} and references therein) because these tasks are at least as hard as BP testing. Such strong separations are exciting because, in addition to helping us understand theoretically what resources separate classical and quantum, they are substantial enough to be observable in near-term experiments (e.g. Ref.~\cite{huang2022Quantum}).

    We note that Ref.~\cite{harrow2023Approximate} derives an upper bound on the bias (and thus a lower bound on the sample complexity) of any PPT product tester for the case of $n=2$ (where BP and MP testers are degenerate). For $n>2$, his argument could be generalized to yield a $\Omega(d^{n/8})$ lower bound on any PPT BP product tester. By restricting our attention to the experimentally-motivated notion of (potentially adaptive) single-copy measurements, we were able to prove a stronger lower bound and avoid the mathematical machinery needed to prove the lower bound against all PPT testers in Ref.~\cite{harrow2023Approximate}. However, it is natural to ask whether our lower bound is optimal. For this, one would need a BP testing algorithm that utilizes only single-copy measurements. To our knowledge, this problem has not yet been studied, thus we leave it as another open question.

    \begin{openquestion}[Single-copy Algorithm for BP Testing]
        Does there exist a BP testing algorithm using at most $O(d^{n/4})$ (potentially adaptive) single-copy measurements? If not, what is the true sample complexity of BP testing in the single-copy regime.
    \end{openquestion}
    
\subsection{Sample Lower Bound for MP Testing with Single-copy Measurements}
Our next result is, to the best of our knowledge, the first non-trivial lower bound on single-copy MP testing algorithms when $n>2$. The proof of the result mirrors that of Theorem~\ref{thm:BP-test-LB}.
\begin{theorem}[Single-copy Lower Bound on MP Testing]\label{thm:MP-test-LB} Any algorithm using, potentially adaptive, single-copy measurements to test whether a state is product across all cuts, or is $\varepsilon$-far from any such state (with probability at least $2/3$) must use at least $\Omega(\sqrt{d/n})$ samples.
\end{theorem}
\begin{proof}
   Recall that from Proposition~\ref{cor:far-from-MP}, an MP tester could be used to distinguish between 
   \begin{align}
        \EC_1 = \{U \ketbra{0}{0} U^{\dagger} : U \sim \UC_{d^n}\} \quad \text{and} \quad \EC_2 = \{\bigotimes_{i=1}^n \psi_i : \psi_i = U_i \ketbra{0}{0} U_i^{\dagger}, U_i \sim \UC_{d}\},
    \end{align}
    with high probability (i.e., $\Omega(1) \leq p_{\rm succ} $). Thus, as in the proof of Theorem~\ref{thm:BP-test-LB}, we can utilize LeCam's two-point method (Lemma~\ref{lem:LeCam-2pt} and Lemma~\ref{lem:one-sided-LeCam}) and triangle inequality. We then simply apply Lemma~\ref{lem:multi-qudit-overlap-LB} and Proposition~\ref{prop:max-mixed-vs-multipartite-prod} to obtain the desired bound. 
\end{proof}
Of course, this bound is only particularly useful when the local dimension is substantially larger than the number of qudits (i.e., when $d \gg n$). We leave the tightening of this bound for future work. In light of the sample complexity upper bound in the next section, we leave the following open question. 

\begin{openquestion}[Lower Bound on Single-copy MP Testers]
    Do there exist two ensembles of states that could be distinguished by an MP tester but which require $\Omega(n\sqrt{d})$ samples when restricted to single-copy measurements? 
\end{openquestion}

\section{Multipartite Product Testing with Single-copy, Local Measurements} \label{sec:MP-testing-alg}
Sample complexity lower bounds are, in a sense, powerful no-go results. In this work, our lower bounds say that there cannot exist single-copy algorithms for these entanglement tasks unless their sample complexity scales with the dimension of the local subsystems. However, it is often of more immediate practical interest if one can construct an algorithm that actually solves a given problem using experimentally feasible measurements. In this section, we provide just such an algorithm for MP testing using single-copy measurements.

The idea is to use the fact that if an $n$-qudit pure state $\ket{\psi}\in(\mathbb{C}^d)^{\otimes n}$ was $\eps_\mathrm{prod}$-far in trace distance from the set of all MP states $\MC_n$, then one would expect that at least one of the reduced qudit states $\psi_i=\mathrm{tr}_{[n]\setminus \{i\}}\brak{\psi}$ is sufficiently impure, i.e., $\exists i$ such that $1-\tr{\psi_i^2}\ge \eps_{\mathrm{rej}}$ for some tolerance $\eps_\mathrm{rej}$. Therefore, we can estimate the purity of each local $\psi_i$ (utilizing a recent algorithm for single-copy purity estimation~\cite{anshu2022Distributed,gong2024sample}) and if they are sufficiently impure then, with some probability of success, we can determine if $\psi\in \MC_n$ or $\eps_\mathrm{prod}$-far from $\MC_n$\footnote{As a quick note: in related works, authors will express quantities in terms of
\begin{align}
   1-\omega = \max_{\varphi \in \MC_n} \abs{\langle \psi | \varphi\rangle}^2
\end{align}
which is related to $\eps_\mathrm{prod}$ via $\eps_\mathrm{prod}=\sqrt{\omega}$. Since we are interested in testing with respect to trace distance, we will express quantities in terms of $\eps_\mathrm{prod}$ over $\omega$.}.

To facilitate this argument, we would like to find such an $\eps_\mathrm{rej}$ which is a function of $\eps_\mathrm{prod}$ and possibly $n$ given that $\ket{\psi}$ is $\eps_\mathrm{prod}$-far from $\MC_n$. In other words, we would like to upper bound how pure some $\psi_i$ can be given that $\psi$ is $\eps_\mathrm{prod}$-far from $\MC_n$. This leads to the following lemma:
\begin{lemma}\label{lem:avg_purity_upperbound}
    Let $\ket{\psi}$ be an $n$-qudit pure state which is $\eps_{\mathrm{prod}}$-far from $\MC_n$ in trace distance. Then,
    \begin{align}
        \frac{1}{n}\sum_i \tr{\psi_i^2} \le 1-\frac{4}{n}\eps_\mathrm{prod}^2(1-\eps_\mathrm{prod}^2),
    \end{align}
    where $\psi\coloneqq \ketbra{\psi}{\psi}$ and $\psi_i=\mathrm{tr}_{[n]\setminus \{i\}}\brak{\psi}$ for $1\le i \le n$. Namely, there always exists $1\le i \le n$ such that $\tr{\psi_i^2}\le 1-\tfrac{4}{n}\eps_\mathrm{prod}^2(1-\eps_\mathrm{prod}^2)$.
\end{lemma}
% \luke{Not sure where to put this note, move as you will.} As a quick note: in related works, authors will express quantities in terms of
% \begin{align}
%    1-\omega = \max_{\varphi \in \MC_n} \abs{\langle \psi | \varphi\rangle}^2
% \end{align}
% which is related to $\eps_\mathrm{prod}$ via $\eps_\mathrm{prod}=\sqrt{\omega}$. Since we are interested in testing with respect to trace distance, we will express quantities in terms of $\eps_\mathrm{prod}$ over $\omega$.
\begin{proof}
    Let $\ket{\psi}$ be an $n$-qudit pure state which is $\eps_{\mathrm{prod}}$-far from $\MC_n$ in trace distance. That is,
    \begin{align}
        \varepsilon_{\mathrm{prod}}=\min_{\varphi\in \MC_n} \frac{1}{2} \norm{\psi - \varphi}_1 \implies \varepsilon^2_{\mathrm{prod}}=1-\left|\braket{\psi|\varphi}\right|^2, \quad \varphi\coloneqq \arg \min_{\varphi \in \MC_n} \frac{1}{2} \norm{\psi-\varphi}_1
    \end{align}
    where $\varphi$ is the nearest MP state to $\psi$. To help with the analysis, we will take $\varphi = \ketbra{0}{0}$ to be the all zero-state where $0\in [d]^n$. We can always do this without loss of generality for the following reason: since $\varphi \in \MC_n$ it follows that $\ket{\varphi}=U_1\ot U_2\ot \cdots \ot U_n \ket{0}=U\ket{0}$ is $\ket{0}$ transformed by local rotations. These local transformations can be absorbed onto $\ket{\psi}$ such that $\ket{0}$ is the closest MP state to $\ket{\psi'}=U^\dagger \ket{\psi}$ at the same distance of $\eps_\mathrm{prod}$. However, since the local purities of $\tr{\psi_i^2}$ are invariant under local rotations we have that $\tr{\psi_i^2}=\tr{\psi_i'^2}$ for all $1\le i \le n$, i.e., local rotations don't change the entanglement structure. Thus, we can replace $\ket{\psi}$ with $\ket{\psi'}$ without changing any aspects of the problem that we care about, namely the distance to $\MC_n$ in trace distance $\eps_\mathrm{prod}$ and the local purities $\tr{\psi_i^2}$. This allows us to write $\ket{\psi}$ in terms of $\ket{0}$ and some state $\ket{\phi}$ which is perpendicular to $\ket{0}$:
    \begin{align}
        \ket{\psi} = \sqrt{1-\eps_{\mathrm{prod}}^2}\ket{0}+\sqrt{\eps_{\mathrm{prod}}^2}\ket{\phi}
    \end{align}
    where $\ket{\phi}=\sum_{x \neq 0} \alpha_x\ket{x}$ is a pure state such that $\langle \phi|0\rangle=0=\alpha_0$. It turns out that, without loss of generality, we can also assume that $\alpha_x=0$ for all $x\in [d]^n$ with Hamming weight $w(x)=|\{x_i \ne 0\}|=1$ (see the proof of Thm. 18 in~\cite{harrow_testing_2010}). This fact will greatly simplify our analysis and the upper bound on the average single-qudit purities. Indeed, consider an arbitrary single-qudit reduced state $\psi_i$:
    \begin{align}
        \psi_i = (1-\eps_\mathrm{prod}^2)\ketbra{0}{0}+\eps_{\mathrm{prod}}^2 \phi_i + \sqrt{\eps_{\mathrm{prod}}^2(1-\eps_{\mathrm{prod}}^2)}\pren{\ketbra{0}{\phi}_i+\ketbra{\phi}{0}_i},
    \end{align}
    where we will employ the notation that for $A\in \LC((\mathbb{C}^d)^{\otimes n})$ the reduced operator on the $i$-th system is $A_i=\mathrm{tr}_{[n]\setminus \{i\}}\brak{A}$. Directly computing $\ketbra{\phi}{0}_i$,
    \begin{align}
        \ketbra{\phi}{0}=\sum_{\substack{x\in [d]^n:\ \\ x\neq 0}} \alpha_x \ketbra{x}{0} \implies \ketbra{\phi}{0}_i&=\sum_{x\neq 0} \alpha_x\mathrm{tr}_{[n]\setminus \{i\}}\brak{\ketbra{x}{0}}=\sum_{x\neq 0}\alpha_x \ketbra{x_i}{0_i} \prod_{j\ne i} \underbrace{\langle 0_j|x_j\rangle}_{\delta_{x_j,0_j}},\\
        &= \sum_{x_i=1}^{d-1} \alpha_{0\cdots 0x_i0 \cdots 0} \ketbra{x_i}{0},\\
        &=0,
    \end{align}
    where the last equality follows from the fact that $\alpha_x=0$ for all $x\in [d]^n$ with $w(x)=1$ which are precisely the type of terms we are left with above. Therefore, $\ketbra{\phi}{0}_i=0=(\ketbra{\phi}{0}_i)^\dagger=\ketbra{0}{\phi}_i$. This is quite nice as it greatly simplifies the expression for the purity of $\psi_i$ to
    \begin{align}
        \tr{\psi_i^2}&=(1-\eps_{\mathrm{prod}}^2)^2 \underbrace{\tr{\ketbra{0}{0}^2}}_{=1}+\eps_{\mathrm{prod}}^4 \underbrace{\tr{\phi_i^2}}_{\le 1}+2\eps_{\mathrm{prod}}^2(1-\eps_{\mathrm{prod}}^2)\bra{0}\phi_i \ket{0},
    \end{align}
   where we can easily bound the first two terms. To bound the overlap of $\phi_i$ with $\ket{0}$, we first compute $\phi_i$,
   \begin{align}
      \phi_i=\sum_{x,y\ne 0}\alpha_x \alpha_y^* \mathrm{tr}_{[n]\setminus \{i\}}\brak{\ketbra{x}{y}} =\sum_{x,y\ne 0}\alpha_x \alpha_y^* \ketbra{x_i}{y_i} \prod_{j\ne i} \underbrace{\langle y_j|x_j\rangle}_{\delta_{y_j,x_j}} = \sum_{j\ne i}\sum_{x_j=y_j=0}^{d-1}\sum_{x_i,y_i=0}^{d-1}\alpha_x \alpha_y^* \ketbra{x_i}{y_i},
   \end{align}
   where the overlap with $\ket{0}$ becomes,
   \begin{align}
       \bra{0}\phi_i \ket{0}= \sum_{j\ne i}\sum_{x_j=y_j=0}^{d-1}\sum_{x_i,y_i=0}^{d-1}\alpha_x\alpha_y^* \underbrace{\langle 0|x_i\rangle}_{\delta_{x_i,0}} \underbrace{\langle y_i|0\rangle}_{\delta_{y_i,0}}=\sum_{x:\ x_i=0} |\alpha_x|^2.
   \end{align}
   Since we're bounding the average $\tfrac{1}{n}\sum_i \tr{\psi_i^2}$, we can instead bound the above quantity summed over all $i$:
    \begin{align}
        \sum_i \sum_{x:\ x_i=0}|\alpha_x|^2 = \sum_x (n-w(x))|\alpha_x|^2=\sum_{w(x)\ge 2}(n-w(x))|\alpha_x|^2 \le \sum_{w(x)\ge 2}(n-2)|\alpha_x|^2=n-2,
    \end{align}
    where in the first equality we use the fact that each $|\alpha_x|^2$ appears with a multiplicity given by the number of zeros in the string $x$ which is simply $n-w(x)$ and in the second equality we used the fact that $\alpha_x=0$ for $w(x)<2$, and in the final equality the fact that $\langle \phi|\phi \rangle=\sum_{w(x)\ge 2}|\alpha_x|^2=1$. In total, we have shown that
    \begin{align}
        \frac{1}{n}\sum_i \tr{\psi_i^2} &\le (1-\eps_\mathrm{prod}^2)^2+\eps_{\mathrm{prod}}^4+\frac{1}{n}2\eps_{\mathrm{prod}}^2(1-\eps_\mathrm{prod}^2)(n-2),\\
        &= 1-2\eps_\mathrm{prod}^2(1-\eps_\mathrm{prod}^2)+\frac{1}{n}2\eps_\mathrm{prod}^2(1-\eps_\mathrm{prod}^2)(n-2),\\
        &=1-\frac{4}{n}\eps_\mathrm{prod}^2(1-\eps_\mathrm{prod}^2),
    \end{align}
    as desired. Note that because this holds on average, there must always exist an $i$ such that $\tr{\psi_i^2}\le 1-\tfrac{4}{n}\eps_\mathrm{prod}^2(1-\eps_\mathrm{prod}^2)$.
\end{proof}

Lemma~\ref{lem:avg_purity_upperbound} allows us to ensure that given a state $\psi$ is $\eps_\mathrm{prod}$-far from $\MC_n$, then there exists a reduced state $\psi_i$ which is sufficiently impure so that with high probability we can detect its deviation from $1$. This leads to the following single-copy product testing algorithm based on single-copy purity estimation from~\cite{anshu2022Distributed,gong2024sample}.

\begin{theorem}\label{thm:MP-test-UB}
    Let $\ket{\psi}\in \left( \mathbb{C}^d \right)^{\otimes n}$ be an $n$-qudit state which is promised to be either in $\MC_n$ or $\frac{1}{\sqrt{2}}\geq \eps_\mathrm{prod}$-far from it. Then, there exists a non-adaptive algorithm utilizing only single-copy, local measurements that requires at most $\OC(n\log(n) d^{1/2}\eps_\mathrm{prod}^{-2})$ copies of $\ket{\psi}$ to determine which is the case with probability at least $2/3$.
\end{theorem}
\begin{comment}
\begin{theorem}
    Suppose there exists an algorithm which can do the following: Given $f_{\rm purity}(d,\varepsilon_{\rm purity},\delta)$ copies of a state $\rho\in \DC(\mathbb{C}^d)$ it outputs a number $\hat{p}(\rho)\in \RB$ such that
    \begin{align}
        \abs{\hat{p}(\rho)-\tr{\rho^2}} < \varepsilon_{\rm purity}
    \end{align}
    with probability of success at least $1-\delta$. Then there exists an algorithm which uses $\OC(n\log{(n)}d^{1/2}\eps^{-1})$ copies of $\ket{\psi}\in (\mathbb{C}^d)^{\otimes n}$ to determine if $\ket{\psi}$ is multipartite product or $\varepsilon_{\rm prod}$-far from multipartite product with probability at least $2/3$. 
\end{theorem}
\end{comment}

\begin{proof}
The idea of the algorithm is that, if a pure state $\psi\coloneqq \ketbra{\psi}{\psi}$ is multipartite product, then all of the marginals are pure as well. Thus, we simply estimate the purity of each reduced single qudit state $\psi_i=\mathrm{tr}_{[n]\setminus \{i\}}\brak{\psi}$ for $1\le i \le n$ to precision $\eps_\mathrm{purity}$ with probability of success at least $1-\delta$. Then, if the estimated purity deviates from $1$ by more than some amount $\eps_\mathrm{rej}$ for any $i$, we have reasonable confidence that the state is far from MP. Using Lemma~\ref{lem:avg_purity_upperbound}, we see that we can take $\eps_\mathrm{purity} < \frac{2}{n}\eps_\mathrm{prod}^2(1-\eps^2_\mathrm{prod})$. For concreteness, we fix $\eps_\mathrm{purity} = \frac{1}{n}\eps_\mathrm{prod}^2(1-\eps^2_\mathrm{prod})$.

Using the algorithm from~\cite{anshu2022Distributed,gong2024sample}, given $\OC(\log(1/\delta)d^{1/2}\eps_\mathrm{purity}^{-1})$ copies of some state $\rho\in \LC(\mathbb{C}^d)$ one can use single-copy measurements to output an estimate $\hat{p}(\rho)\in \mathbb{R}$ such that
\begin{align}
    \abs{\hat{p}(\rho)-\tr{\rho^2}}\le \eps_\mathrm{purity} \quad \text{w/ probability at least $1-\delta$}.
\end{align}
Naively, one could simply do this protocol $n$ times to output estimates for each $\tr{\psi_i^2}$ using single-copy local measurements. However, this protocol can actually be done in parallel, i.e., estimate the purities $\tr{\psi_i^2}$ simultaneously, since the protocol only needs independence between copies of $\psi_i$ on different measurement rounds not independence between $\psi_i$ and $\psi_j$ on the same measurement round. Doing so will avoid an extra factor of $n$ in the sample complexity. Thus, the proposed algorithm is to do the following:
\begin{algorithm}[H]
\caption{Multipartite Product Test with Single-copy Measurements}\label{alg:MP-tester}
    \begin{algorithmic}[1]  
        \State \textbf{Input:} $\ket{\psi}\in (\mathbb{C}^d)^{\otimes n}$ with $\psi\coloneqq \ketbra{\psi}{\psi}$ such that either $\psi \in \mathcal{M}_n$ or $\psi$ is $\eps_\mathrm{prod}$-far from $\mathcal{M}_n$.

        \State $\eps_\mathrm{purity} \gets \frac{1}{n}\eps_\mathrm{prod}^2(1-\eps_\mathrm{prod}^2)$.

        \State $\delta \gets \frac{1}{3n}$.

        \State Collect $\OC(\log(1/\delta) d^{1/2} \eps_\mathrm{purity}^{-1}) = \OC(n\log n\ d^{1/2} \eps_\mathrm{prod}^2)$ copies of $\psi$.
        
        \State \textbf{Estimate:} Using the copies of $\psi$, compute $\{\hat{p}(\psi_i)\}_{i=1}^n\subset \mathbb{R}$ such that $\abs{\hat{p}(\psi_i)-\tr{\psi_i}^2}\le \eps_\mathrm{purity}$ with probability at least $1-\delta$ for each $1\le i \le n$.
        \For{$1\le i \le n$}
        \If{$\hat{p}(\psi_i) \leq 1-2\eps_\mathrm{purity} $}
        \State \textbf{Output:} \texttt{Reject} 
        \EndIf
        \EndFor
        \State \textbf{Output:} \texttt{Accept} 
    \end{algorithmic}
\end{algorithm}

Say that $\psi \in \MC_n$, then the algorithm successfully accepts if all estimators return an a value $\eps_\mathrm{purity}$-close to the true subsystem purities of $1$:
\begin{align}
    \Pr[\texttt{accept} \ \vert \ \psi \in \mathcal{M}_n] & = \Pr[\hat{p}(\psi_i) > 1-2\eps_\mathrm{purity} \ \forall i]\\
    & \geq \Pr[ \vert \hat{p}(\psi_i) - \tr \psi_i^2 \vert < \eps_\mathrm{purity} \ \forall i \vert]\\
    & \geq 1 - \sum_i \Pr[ \vert \hat{p}(\psi_i) - \tr \psi_i^2 \vert \geq \eps_\mathrm{purity}]\\
    & \geq 1 - n\cdot \frac{1}{3n} = \frac{2}{3}\ .
\end{align}

If $\psi$ is $\eps_\mathrm{prod}$ far from $\MC_n$, then Lemma~\ref{lem:avg_purity_upperbound} guarantees that there is a subsystem $i^*$ such that $\Tr[\psi_{i^*}^2] \leq 1-4\eps_\mathrm{purity}$. Hence, if the purity of this subsystem is successfully estimated to precision $\eps_\mathrm{purity}$, then the algorithm correctly rejects:
\begin{align}
    \Pr[\mathrm{accept} \ \vert \ \text{$\psi$ $\eps_\mathrm{prod}$- far from $\MC_n$}] & = \Pr[\hat{p}(\psi_i) \geq 1-2\eps_\mathrm{purity} \ \forall i ]\\
    & \leq \Pr[ \hat{p}(\psi_{i^*}) \geq 1-2\eps_\mathrm{purity}]\\
    & \leq \frac{1}{3n}\ .
\end{align}
This completes the proof.
\end{proof}
We note that the probability of passing the product test was shown, Lemma 2 of Ref.~\cite{harrow_testing_2010}, to be 
\begin{align}
    P_{\rm test}(\rho) = \frac{1}{2^n} \sum_{S \subseteq [n]} \tr{\rho_S^2}.
\end{align}
Thus, if one could estimate this functional to sufficient precision, it could be used to determine if a state is MP or not. It is interesting to point out that the sample complexity of Algorithm~\ref{alg:MP-tester} is exponentially better than the best-known single-copy, local algorithm~\cite{coffman2024Local} for the task of estimating this functional. While this estimation must be at least as hard as testing, it is not a priori obvious that the estimator wouldn't yield a near-optimal tester as in the case of purity testing/estimation where the testing lower bound from Ref.~\cite{chen2022Exponential} is nearly saturated by the purity estimation algorithms in Ref.~\cite{anshu2022Distributed,gong2024sample} (at least for constant $\varepsilon$). At the time of writing, we suspect our algorithm is optimal up to log factors, though we leave this as a final open question.

\begin{openquestion}[Upper Bound on the Sample Complexity of MP Testing]
    Can the Algorithm~\ref{alg:MP-tester} be improved yield an $O(n \sqrt{d} \varepsilon^{-2})$ upper bound on the sample complexity of MP testing in the single-copy regime?
\end{openquestion}

\begin{comment}

\end{comment}

\section{Conclusions and Future Directions}
In this work, we establish an exponential separation between BP testing algorithms with and without the ability to perform multi-copy measurements. Our lower bound on BP testers implies hardness results for estimating certain multipartite entanglement measures~\cite{jones2024Testing} and for locating unentanglement~\cite{bouland2024state} using only single-copy measurements. We provide the first non-trivial lower bound on MP testers for the case when $n>2$ before providing a single-copy, local algorithm for MP testing that is optimal when $n=2$. Along the way, we highlight several interesting open problems for future directions, all of which may help in understanding separations (or lack thereof) between single-copy, local and single-copy, global measurement strategies. Such separations would have interesting theoretical implications and immediate relevance to near-term experiments.

\subsection*{Acknowledgments} 
The authors would like to acknowledge helpful discussions with Ryan O'Donnell, Fernando Geronimo, Qizhao Huang, and Marius Junge. J.L.B thanks Shawn Geller for making him aware of Ref.~\cite{grone1988Permanental}. J.L.B is supported by a National Science Foundation Mathematical Sciences Postdoctoral Research Fellowship under Award No.~2402287 as well as an IQUIST Postdoctoral Fellowship. L.C. is supported by the National Science Foundation Graduate Research Fellowship under Grant No. 2140743. L.S. is supported by IBM through the Illinois-IBM Discovery Accelerator Institute. F.L. is supported by the National Science Foundation under Grant No.~2442410. Any opinions,
findings, and conclusions or recommendations expressed in this material are those of the
author(s) and do not necessarily reflect the views of the National Science Foundation.

\bibliographystyle{alphaurl}
{\small \bibliography{main.bib}}

\end{document}